\documentclass{tlp}

\usepackage{listings}
\usepackage{graphicx}
\usepackage[all]{xy}
\usepackage{eepic,latexsym,amssymb}
\usepackage{fancybox,epic}

\newtheorem{definition}{Definition} % [section]
\newtheorem{example}{Example} % [section]
\newtheorem{proposition}{Proposition} % [section]
\newtheorem{theorem}{Theorem} % [section]

\newcommand{\AS}{\mbox{\it AS}}
\newcommand{\marker}{\ \uparrow\ }
\newcommand{\pop}{\mbox{\sf pop}}
\newcommand{\push}{\mbox{\sf push}}
\newcommand{\createempty}{\mbox{\sf empty}}
\newcommand{\contents}{\mbox{\sf contents}}
\newcommand{\eval}[2]{(#1,#2)}
\newcommand{\exec}{\mbox{\sf Exec}}
\newcommand{\emptysubs}{\mbox{\sf id}}

\newcommand{\short}[1]{}
\newcommand{\finalshort}[1]{}

\usepackage{textcomp}
\newcommand{\N}{\mbox{$I\!\!N$}}

\newcommand{\selatpar}[1]{A_{R#1}}
\newcommand{\selat}{A_{R}}

\newcommand{\cR}{{\cal R}}

\long\def\comment#1{}

\def\tuple#1{\langle{#1}\rangle}
\newcommand{\gd}[0]{~\rule{0.5mm}{2.5mm}~}

\newcommand{\ciao}{\texttt{Ciao}}
\newcommand{\ciaopp}{\texttt{CiaoPP}}

\newcommand{\secbeg}{}
\newcommand{\secend}{}

\newcommand{\select}{\cR}
\newcommand{\simplify}{\mathit{simplify}}

\begin{document}

\title[Efficient Local Unfolding with Ancestor Stacks] %for Full Prolog
{Efficient Local Unfolding with Ancestor Stacks\thanks{A preliminary
    version of this work appeared in the Post-proceedings of
    LOPSTR'04, LNCS 3573, Springer-Verlag, 2005.}
       %for Full Prolog
}
 
\author[G. Puebla and E. Albert and M. Hermenegildo]
{GERM\'AN PUEBLA$^{1}$ \and ELVIRA ALBERT$^{2}$ \and MANUEL HERMENEGILDO$^{1,3}$ \vspace*{3mm} \\ 
$^{1}$School of Computer Science, Technical University of Madrid\\
E28660-Boadilla del Monte, Madrid, Spain. E-mail: {\tt
  german@fi.upm.es}, {\tt herme@fi.upm.es}\vspace*{2mm} \\
$^{2}$School of Computer Science, Complutense University of Madrid \\ 
E28040-Profesor Jos\'e Garc\'{\i}a Santesmases, s/n, Madrid, Spain. E-mail: {\tt
  elvira@sip.ucm.es} \vspace*{2mm} \\
$^{3}$Madrid Institute for Advanced Studies in Software Development Technology,\\
IMDEA Software. E-mail: {\tt manuel.hermenegildo@imdea.org}
% Departments of Computer Science and Electrical and Computer 
% Engineering\\ 
% University of  New Mexico. E-mail: {\tt herme@unm.edu}
}

%% \author{Germ\'{a}n Puebla\inst{1} \and Elvira Albert\inst{2}
%%  \and Manuel Hermenegildo\inst{1,3}
%% }

%% \institute{
%% School of Computer Science, Technical U.~of Madrid, 
%% \email{\{german,herme\}@fi.upm.es}
%% \and
%% School of Computer Science, Complutense U.~of  Madrid,
%% \email{elvira@sip.ucm.es}
%% \and
%% Depts. of Comp.~Sci.~and El.~and Comp.~Eng., U. of  New Mexico,
%% \email{herme@unm.edu}
%% } 

%% \pagerange{\pageref{firstpage}--\pageref{lastpage}}
%% \volume{\textbf{10} (3):}
%% \jdate{March 2002}
\setcounter{page}{1}
%\pubyear{2002}

\maketitle

\label{firstpage}
\maketitle
\begin{abstract} 
  
%   In spite of the important research efforts in the area, the
%   integration of powerful partial evaluation methods into practical
%   compilers for logic programs is still far from reality.  This is
%   related both to 1) efficiency issues and to 2) the complications of
%   dealing with practical programs.  Regarding the first issue, t
  The
  most successful \emph{unfolding rules} used nowadays
  in the partial evaluation of logic programs 
  are based on %structural orders
%  \emph{well founded orders} (wfo) or 
\emph{well quasi orders} (wqo)
  applied over (covering) \emph{ancestors}, i.e., a subsequence of the
  atoms selected during a derivation.  Ancestor (sub)sequences are
  used to increase the specialization power of unfolding while still
  guaranteeing termination and also to reduce the number of atoms for
  which the  wqo has to be checked.  Unfortunately, maintaining
  the structure of the ancestor relation during unfolding introduces
  significant overhead.
  We propose an efficient, practical \emph{local} unfolding rule based
  on the notion of covering ancestors which can be used in combination
  with a wqo and allows a stack-based implementation without losing
  any opportunities for specialization. Using our technique, certain
  non-leftmost unfoldings are allowed as long as local unfolding is
  performed, i.e., we cover depth-first strategies.
%  Regarding %% the second issue,
To  deal with practical programs, we propose assertion-based
% and global analysis-based 
techniques which allow our approach to treat programs that 
include (Prolog) built-ins and external predicates in a very
extensible manner, for the
  case of leftmost unfolding.
%% \sidenote{This should probably go now.}
%%   We outline as well a number of generalizations that allow
%%   certain non-leftmost unfoldings.
  %   
  Finally, we report on our implementation of these techniques
  embedded in a practical partial evaluator, which shows that our
  techniques, in addition to dealing with practical programs, are also
  significantly more efficient in time and somewhat more efficient in
  memory than traditional tree-based implementations.
To appear
in Theory and Practice of Logic Programming (TPLP).

%  in a
%   practical partial evaluator, embedded in a state of the art compiler
%   which uses global analysis extensively. % : the \ciao\ compiler and,
% %   specifically, its preprocessor \ciaopp
%     % 
%   The performance analysis of the resulting system shows that our
%   techniques, in addition to dealing with practical programs, are also
%   significantly more efficient in time and somewhat more efficient in
%   memory than traditional tree-based implementations.\\ %\vspace*{3mm}
  % 
%  We believe that our approach contributes to the practicality of
%  state-of-the-art partial evaluation techniques.
%%   based on 
%%   well-known and useful
%%   orderings, such as \emph{homeomorphic embedding}.

\submitted {25 October 2005}
 \revised {3 March 2006, 2 July 2009, 14 October 2009}
 \accepted{17 November 2009}

\vspace*{-3mm}
\end{abstract}

\secbeg
\section{Introduction}
\label{sec:introduction}
\secend

The main purpose of \emph{partial evaluation} (see \cite{pevalbook93}
for a general text on the area) is to specialize a given program
w.r.t.\ part of its input data---hence it is also known as
\emph{program specialization}.  Essentially, partial evaluators are
non-standard interpreters which evaluate expressions while enough
information is available and residualize them otherwise.  The partial
evaluation of logic programs is usually known as \emph{partial
  deduction} \cite{Lloyd:jlp91,gallagher:pepm93}.  Informally, the
partial deduction algorithm proceeds as follows. Given an input
program and a set of atoms, the first step consists in applying an
\emph{unfolding rule} to compute finite (possibly incomplete) SLD
trees for these atoms. This step returns a set of \emph{resultants}
(or residual rules), i.e., a program, associated to the root-to-leaf
derivations of these trees.  Then, an \emph{abstraction operator} is
applied to properly add the atoms in the bodies of resultants to the
set of atoms to be partially evaluated. The abstraction phase yields a
new set of atoms, some of which may in turn need further evaluation
and, thus, the process is iteratively repeated while new atoms are
introduced.
% The final set of 
% terms unambiguously determines the associated partial evaluation.
% 
The number of such new atoms which can be introduced can in general be
unbounded.
%In order for the partial deduction to always terminate a
%  number of additional rules are needed to control the
%  process.
%
  The termination of the partial deduction process is ensured by two
  control issues.  Following the terminology of
  \cite{gallagher:pepm93}, the so-called \emph{local} control defines
  an unfolding rule which determines how to construct finite SLD
  trees. The \emph{global} control defines an abstraction operator
  which guarantees that the number of new atoms is kept finite.
  Termination of the partial deduction algorithm involves ensuring
  termination both at the local and global levels. We refer to
  \cite{LeuschelBruynooghe:TPLP02} for a survey on both control
  issues.
  This article is centered on the local control, namely on the
  development of a practical, efficient unfolding rule. The techniques
  we will propose for local control can be used in combination with
  any global control strategy.

%   In spite of the important research efforts in the area, the
%   integration of partial deduction methods into compilers seems to be
%   still far from reality.  

  We believe that two factors limiting the general uptake of
  partial deduction are: 1) the
  relative inefficiency of the partial deduction method, and 2) the
  complications brought about by the treatment of real programs.
  Indeed, the integration of powerful strategies in the unfolding rule
  ---like the use of wqos combined with the
  ancestor relation--- can introduce a significant cost both in time
  and memory consumption of the specialization process. Regarding the
  treatment of real programs which include external predicates,
  non-declarative features, etc., the complications range from how to
  identify which predicates include these non-declarative features
  (ad-hoc but difficult to maintain tables are often used in practice
  for this purpose) to how to deal with such predicates during partial
  deduction.  Also, the optimal treatment of these predicates during
  partial deduction often requires information which can only be
  available at partial deduction time if a global analysis of the
  program is performed.  
%   A main objective of this work is to
%   contribute to the uptake of partial evaluation techniques by
%   proposing novel solutions to some of these issues.
  Our main objective in this work is to propose some novel solutions
  to these issues.
  
  State-of-the-art partial evaluators integrate terminating unfolding
  rules for local control based on wqos, like
  homeomorphic embedding \cite{Kru60,LeuschelBruynooghe:TPLP02} which can
  obtain very powerful optimizations.
  Moreover, they allow performing the ordering comparisons over
  \emph{subsequences} of the full sequence of the selected atoms.
%Regarding the efficiency of the partial deduction process, 
  In particular, the use of \emph{ancestors} for refining sequences of
  visited atoms, originally proposed in~\cite{BSM92}, greatly improves
  the specialization power of unfolding while still guaranteeing
  termination and also reduces the length of the sequences for which
  the embedding order for the new atoms has to be checked.
\short{ However, although an important amount of effort has been
  devoted to improving the implementation of ancestors-based
  techniques (see, e.g.,~\cite{MartensDeSchreye:jlp95b}) current
  state-of-the art partial deduction systems (e.g., the {\sf Ecce} partial deduction
  system or the Curry partial evaluator \cite{ElviraHanusVidal02JFLP})
  are still not very efficient when local control based on the
  combination of wqo and ancestors is used.
} Unfortunately, having to maintain dependency information for the
individual atoms in each derivation during the generation of SLD trees
has turned out to introduce overheads which seem to cancel out the
theoretical efficiency gains expected.
% from using ancestors.
%
In order to address this issue,
in this article,
%
% we introduce \emph{ASLD resolution} as the basis for a novel unfolding
% rule which relies on the notion of covering ancestors and which allows
% a very efficient implementation technique based on stacks.  
we introduce \emph{ASLD resolution} as the basis for an 
efficient, stack-based implementation technique of a local
unfolding rule relying on the notion of covering ancestors.  Our
technique can significantly reduce the overhead incurred by the use of
covering ancestors without losing any opportunities for
specialization. We outline as well a generalization that allows
certain non-leftmost unfoldings with the same assurances.

In order to deal with real programs that include (Prolog) built-ins
and external predicates, we extend ASLD resolution 
and the ancestor-based local unfolding rule 
to handle these 
predicates by relying on assertion-based techniques
\cite{assert-lang-disciplbook-short}. The use of assertions provides
\emph{extensibility} in the sense that users and developers of partial
evaluators can deal with new external predicates during partial
evaluation by just adding the proper assertions to these predicates
---without having to maintain ad-hoc tables or modifying the partial
evaluator itself.
% and global
% analysis-based techniques.
%
\short{
 which
allow our approach to deal with real programs that include (Prolog)
built-ins and external predicates, for the case of leftmost unfolding.
}
%
%% We outline as well a number of generalizations that allow certain
%% non-leftmost unfoldings.  
We report on an implementation of our technique in a practical,
state-of-the-art partial evaluator, embedded in a production compiler
which uses assertions and global analysis extensively (the \ciao\ 
compiler \cite{ciao-manual-1.10-short} and, specifically, its
preprocessor \ciaopp\ \cite{ciaopp-sas03-journal-scp}).  We believe
that our experimental results  provide evidence that
our technique pays off in practice and can thus contribute to the
practicality of state-of-the-art partial evaluation techniques.

An important observation is that the techniques that we propose in this
article to control the unfolding process are useful in the context of
\emph{online} partial evaluation.  Traditionally, two approaches to
partial evaluation have been considered, \emph{online} and
\emph{offline} partial evaluation (see
\cite{LCBV04,LeuschelBruynooghe:TPLP02}).
In online partial evaluation all control decisions are taken on the
fly during the specialization phase, by keeping track of the
specialization history (e.g., the ancestor subsequences). 
In the offline approach, all control decisions are taken before the
specialization phase proper. These control decisions are based on
abstract descriptions of the data instead of the actual data. The
control strategy is usually represented as program annotations which
are the sole decision criteria for control of the partial
evaluator. For instance, regarding local control, an annotation can
explicitly indicate that an atom should not be unfolded.  Regarding
global control, annotations typically specify for each call which
arguments have to be generalised away (i.e., replaced by variables).
Such annotations are generated automatically in some partial
evaluators by a \emph{binding-time
  analysis}~\cite{DBLP:conf/lopstr/CraigGLH04}, while in other partial
evaluators they are manually provided by the user, either in part or
in full.  The advantages of the offline approach are that, once all
control annotations are available, partial evaluation is quite simple
and efficient. On the other hand, online partial evaluation while
usually less efficient, it tends to have more powerful control
strategy since control decisions are based on actual data instead of
abstract descriptions of data.  In principle, one could argue that
both approaches are equally powerful (see
\cite{DBLP:journals/toplas/ChristensenG04}) and that the offline
approach can be more appropriate if the output of a global program
analysis is available, while online partial evaluators usually
only consider local, runtime information.
% Therefore, though all
% offline partial evaluations can be replicated using online techniques,
% many online partial evaluations cannot be reproduced using offline
% techniques.
In this work, we are interested in proposing novel
techniques which help improve the efficiency of online partial
evaluation.

The structure of the article is as follows.
Section~\ref{sec:background} presents some required background on
local control during partial deduction.
Section~\ref{sec:usefulness-ancestors} shows by means of an example
why using ancestors is needed.  Section~\ref{sec:an-effic-impl}
presents ASLD resolution as the basis for an efficient unfolding rule
based on ancestors which allows a stack-based implementation.
Section~\ref{sec:an-unfolding-rule} extends the unfolding techniques
to the case of external predicates.
% Section~\ref{sec:comp-rules-part}
% discusses different choices available when deciding the computation rule to
%use.
Section~\ref{sec:experiments} presents some experimental results which
compare the performance of different unfolding strategies with several
implementations.  Finally, Section~\ref{sec:disc-future-work}
discusses some related work and concludes.

\secbeg
\section{Background}
\secend
\label{sec:background}

We assume some basic knowledge on the terminology of logic
programming. See for example~\cite{Lloyd87} for details.

Very briefly, an \emph{atom} $A$ is a syntactic construction of the
form $p(t_1,\ldots,t_n)$, where $p/n$, with $n\geq 0$, is a predicate
symbol and $t_1,\ldots,t_n$ are terms. The function $pred$ applied to
atom $A$, i.e., $pred(A)$, returns the predicate symbol $p/n$ for $A$. A
\emph{clause} is of the form $H\leftarrow B$ where its head $H$ is an
atom and its body $B$ is a conjunction of atoms. A \emph{definite
  program} is a finite set of clauses. A \emph{goal} (or query) is a
conjunction of atoms.

We denote by $\{X_1 \mapsto t_1,\ldots, X_n \mapsto t_n\}$ the
\emph{substitution} $\sigma$ with
$\sigma(X_i) = t_i$  for $i=1,\ldots,n$ (with $X_{i}\neq 
X_{j}$ if  $i\neq 
j$), and $\sigma(X) = X$ for all other variables $X$.
Given an atom $A$, $\theta(A)$ denotes the application of substitution
$\theta$ to $A$. Given two substitutions $\theta_1$ and $\theta_2$, we
denote by $\theta_1\theta_2$ their composition. 
The identity substitution is denoted by $id$.

A term $t'$ is an \emph{instance} of $t$ if there is a substitution
$\sigma$ with $t' = \sigma(t)$.

\secbeg
\subsection{Basics of partial deduction}
\secend
\label{sec:basics-part-deduct}

The concept of \emph{computation rule} is used to select an atom
within a goal for its evaluation.

\begin{definition}[computation rule] A \emph{computation rule} is a
  function
%$R:{\cal G} \mapsto {\cal A}$. 
  $\select$ from goals to atoms.  Let $G$ be a goal of the form
  $\leftarrow A_1,\ldots,\selat,\ldots,A_k$, $k\geq 1$.
%
  %% This first A_n gave me a weird output, at least on my LaTeX: was
  %% printed as ``A\''. I really do not understand why. --MH
  If $\select(G)=$$\selat$ we say that $\selat$ is the \emph{selected} atom
  in $G$.
  
\end{definition}
The operational semantics of definite programs is based on
derivations.
\begin{definition}[derivation step]
\label{def:der-step}
Let $G$ be $\leftarrow A_1,\ldots,\selat,\ldots,A_k$. Let $\select$ be a
computation rule and let $\select(G)=$$\selat$. Let $C=H \leftarrow
B_1,\ldots,B_m$ be a renamed apart clause in $P$. Then $G'$ is
\emph{derived} from $G$ and $C$ via $\select$ if the following conditions
hold:
  \begin{eqnarray*}
\theta=mgu(\selat,H) \\
G' \mbox{ is the goal } \leftarrow
  \theta(B_1,\ldots,B_m,A_1,\ldots,\selatpar{-1},\selatpar{+1},\ldots,A_k)
  \end{eqnarray*}
\end{definition}
The definition above differs from standard formulations (such as that
in~\cite{Lloyd87}) in that the atoms newly introduced in $G'$ are not
placed in the same position where the selected atom $\selat$ used to
be, but rather they are placed to the left of any atom in $G$. For
definite programs, this is correct since goals are conjunctions, which
enjoy the commutative property.  This modification will become
instrumental to the operational semantics we propose in forthcoming
sections. This is not true though for programs with extra logical
predicates, as we will discuss in Section~\ref{sec:an-unfolding-rule}.
Also, it is well-known that changing the atom's positions might not
preserve finite failure.  Although our general notion of resolution
allows reordering the atoms, in a practical system, we can allow only
leftmost unfolding and still obtain significant improvements (as
will be explained in Section 5).

As customary, given a program $P$ and a goal $G$, an \emph{SLD
  derivation} for $P\cup\{G\}$ consists of a possibly infinite
sequence $G=G_0, G_1, G_2,\ldots$ of goals, a sequence
$C_1,C_2,\ldots$ of properly renamed apart clauses of $P$, and a
sequence of \emph{computed answer substitutions}
$\theta_1,\theta_2,\ldots$ (or mgus) such that each $G_{i+1}$ 
is derived from $G_i$ and $C_{i+1}$ using $\theta_{i+1}$. If $G_i$ is
of the form $\leftarrow A_1,\ldots,\selat,\ldots,A_k$ and
$G_{i+1}\equiv
\theta(B_1,\ldots,B_m,A_1,\ldots,\selatpar{-1},\selatpar{+1},\ldots,A_k)
$ is derived from $G_i$ (as stated in Definition~\ref{def:der-step}),
we say that each atom $\theta(A_i)$ with
$i=1,\ldots,\sc{R}-1,\sc{R}+1,\ldots,k$ is the \emph{instance
  originating} from $A_i$.  Finally, we say that the SLD derivation is
composed of the \emph{subsequent} goals $G_0,G_1,G_2,\ldots$.

A derivation step can be non-deterministic when $\selat$ unifies with
several clauses in $P$, giving rise to
several possible SLD derivations for a given goal.  Such SLD
derivations can be organized in \emph{SLD trees}.
A finite derivation $G=G_0, G_1, G_2,\ldots, G_n$ is called
\emph{successful} if $G_n$ is empty.  In that case
$\theta=\theta_1\theta_2\ldots\theta_n$ is called the \emph{computed answer}
for goal $G$.  Such a derivation is called \emph{failed} if it is not
possible to perform a derivation step with $G_n$.

In order to compute a partial deduction \cite{Lloyd:jlp91}, given an
input program and a set of atoms, the first step consists in
applying an \emph{unfolding rule} to compute finite (possibly
incomplete) SLD trees for these atoms. Then, a set of
\emph{resultant}s or residual rules are systematically extracted from
the SLD trees.\footnote{Let us note that the definition of a partial
  deduction \emph{algorithm} requires, in addition to an unfolding rule, the
  so-called global control level (see
  Section~\ref{sec:introduction}).}
\begin{definition}[unfolding rule] Given an atom $A$, 
  an \emph{unfolding rule} computes a set of finite SLD derivations
  $D_1,\ldots,D_n$ (i.e., a possibly incomplete SLD tree) of the form
  $D_i=A,\ldots,G_i$ with (a composed) computed answer substitution $\theta_i$ for
  $i=1,\ldots,n$ whose associated \emph{resultants} are $\theta_i(A)
  \leftarrow G_i$.
%% Formally, an unfolding rule $\cU$ can be defined as a function
%% from atoms to atoms such that, whenever $\cU(G) = G'$, then there
%% exists a \emph{finite} SLD derivation $G,\ldots, G'$ with mgu
%% $\theta_1,\ldots,\theta$ whose associated \emph{resultant} (or
%% residual rule) is $\theta(G) \leftarrow G'$.
 \end{definition}
 A \emph{partial evaluation} for the initial atom is then
 defined as the set of resultants, i.e., a program, associated to the
 root-to-leaf derivations for the computed SLD tree.  The partial
 evaluation for a set of atoms is defined as the union of the partial
 evaluations for each atom in the set.  We refer to
 \cite{LeuschelBruynooghe:TPLP02} for details.  \secbeg

\subsection{Termination of local control}
\label{sec:term-local-contr}
\secend

In order to ensure the local termination of the partial deduction algorithm while
producing useful specializations, the unfolding rule must incorporate
some non-trivial mechanism to stop the construction of SLD trees.
Nowadays, %  well-founded orderings (wfo)
% \cite{BSM92,MartensDeSchreye:jlp95b} and
well-quasi orderings (wqo)
\cite{SG95,Leuschel:SAS98} are broadly used in the context of on-line
partial evaluation techniques.  

It is well known that the use of wqos allows the definition
of \emph{admissible} sequences which are always finite.
%
%Intuitively, derivations are expanded as long as there is some
%evidence that interesting computations are performed but also
%guaranteed to terminate (according to the selected ordering).  
%
Intuitively, a sequence of elements $s_1,s_2,\ldots$ in $S$ is called
\emph{admissible with respect to an order $\leq_S$} \cite{BSM92} iff
there are no $i < j$ such that $s_i \leq_S s_j$.  
The next
definition captures this idea.

\begin{definition}[admissible --wqo]\label{def:adm}
  Let $(A_1,\ldots,A_n)$ be a sequence of atoms and $A$ be a new atom
  to be added to the sequence. Let $\leq_S$ be a wqo. We denote by
  $Admissible(A,(A_1,\ldots,A_n),\leq_S)$, with $n\geq 0$ the truth
  value of the expression $\forall A_i,\; i\in\{1,\ldots,n\}\;:
  A \not\geq_S A_i$.
\end{definition}
Given a derivation $G_{1}, G_{2},~ \cdots, G_{n+1}$ in order to decide
whether to evaluate $G_{n+1}$ or not, we check that the selected atom
in $G_{n+1}$ is not strictly greater or equal to any previous
\emph{comparable} selected atom \cite{Leuschel:NeilFest02}. Observe
that the ancestor test is only applied on comparable atoms,
i.e., ancestor atoms with the same predicate symbol.  This corresponds
to the original notion of covering ancestors \cite{BSM92}.  Note that
$A_1,\ldots,A_n$ in the above definition refer to the selected atoms
in $G_1,\ldots,G_n$ and $A$ refers to the selected atom in $G_{n+1}$.
%
% In wfo, it is sufficient to verify that the selected atom is strictly
% smaller than the previous comparable one (if one exists). Let $<$ be a
% wfo, by $Admissible(A,(A_1,\ldots,A_n),<)$, with $n\geq 0$, we denote
% the truth value of the expression $A< A_n$ if $n\geq 1$ and $true$ if
% $n=0$.

% We will denote by \emph{structural order} a wfo or a wqo (written as
% $\leq_S$ to represent any of them).  
Among the 
wqo, the \emph{homeomorphic embedding} ordering
\cite{Kru60} 
has proved to be very powerful in
practice. 
%  The
% interested reader is referred to Leuschel's work \cite{Leuschel:SAS98}
% where a detailed description of homeomorphic embedding can be found.
We recall the definition of homeomorphic embedding, which can be found
for instance in Leuschel's work~\cite{Leuschel:SAS98}.

\newcommand{\emb}{{\scriptsize \trianglelefteq}}

\begin{definition}[$\emb$] Given two  atoms $A=p(t_1,\ldots,t_n)$
  and $B=p(s_1,\ldots,s_n)$, we say that $B$ \emph{embeds} $A$, written
  $A \;\emb \; B$, if $t_i \; \emb \; s_i$ for all $i$ s.t. $1 \leq i \leq n$.
  The embedding relation over  terms, also written $\emb$, is defined by the
  following rules:
\begin{enumerate}
\item $Y\! \; \emb  \; X$ for all variables $X,Y$.
\item $s\! \; \emb\; f(t_1,\ldots,t_n)$ if $s$
  $\emb$ $t_i$ for some $i$.
\item $f(s_1,\ldots,s_n)\! \; \emb  \; f(t_1,\ldots,t_n)$ if
  $s_i \; \emb\; t_i$ for all $i$, $1 \leq i \leq n$.
\end{enumerate}
\end{definition}

Informally, atom $t_{1}$ \emph{embeds} atom $t_{2}$ if $t_{2}$ can be
obtained from $t_{1}$ by deleting some operators, e.g., 
${\tt \underline{f}(g(\underline{A},B),
\underline{h}(\underline{C},s(\underline{D}))}$
embeds 
${\tt f(A,h(C,D))}$.

% $\tt s(\underline{\tt s}(\underline{\tt U}+W ) \underline{\times}
% (\underline{\tt U +} s(\underline{\tt V}))) $ embeds $\tt s(U \times
% (U + V))$.

\subsection{Covering ancestors}
State-of-the-art unfolding rules allow performing ordering comparisons
over \emph{subsequences} of the full sequence of the selected atoms of
a derivation by organizing atoms in a \emph{proof tree} \cite{bruy91},
achieving further specialization in many cases while still
guaranteeing termination.  To do so, they maintain dependencies over
the selected atoms which are chosen in such a way that only a
subsequence of such selected atoms needs to be considered. The essence
of the most advanced techniques is based on the notion of
\emph{covering ancestors} \cite{BSM92}.

\begin{definition}[ancestor relation]
\label{def:ancestor-rel}
Given a derivation step and $\selat$, $B_i$, $i=1,\ldots,m$ as in
Definition~\ref{def:der-step}, we say that $\selat$ is the
\emph{parent} of the instance of $B_i$, $i=1,\ldots,m$, in the goal
and in each subsequent goal where the instance originating from $B_i$
appears. The \emph{ancestor} relation is the transitive closure of the
parent relation.
\end{definition}
The important observation is that a derivation can contain selected
subgoals which are indeed part of a different branch in the proof
tree.  

Given an atom $A$ and a derivation $D$, we denote by $Ancestors(A,D)$
the sequence of (comparable) ancestors of $A$ in $D$ as defined in
Definition~\ref{def:ancestor-rel}. It captures the dependency relation
implicit within a \emph{proof tree}.

It has been proved~\cite{BSM92} that any infinite derivation must have
at least one inadmissible \emph{covering ancestor} sequence, i.e., a
subsequence of the atoms selected during a derivation. Therefore, it
is sufficient to check the selected ordering relation $\leq_S$
over the covering ancestor subsequences in order to detect
inadmissible derivations. 

\begin{definition}[safe step]
  An SLD step is \emph{safe} with respect to a  wqo if
  the covering ancestor sequence of the selected atom is admissible
  with respect to that order.
\end{definition}
The above definition is extended to derivations as follows.
\begin{definition}[safe derivation]\label{def:safe-deriv}
  An SLD derivation is \emph{safe} with respect to a wqo if all
  covering ancestor sequences of the selected atoms are admissible
  with respect to that order.
\end{definition}
Otherwise, the SLD derivation is considered \emph{unsafe}.

\secbeg

\section{The Usefulness of Ancestors}
\label{sec:usefulness-ancestors}
\secend

\begin{figure}[t]
%\footnotesize %  \centering
\begin{minipage}{0.4\textwidth}
\begin{verbatim}
qsort([],R,R).
qsort([X|L],R,R2) :- 
   partition(L,X,L1,L2),
   qsort(L2,R1,R2),
   qsort(L1,R,[X|R1]).
\end{verbatim}
\end{minipage}
 \begin{minipage}{0.55\textwidth}
 \begin{verbatim}
 partition([],_,[],[]).
 partition([E|R],C,[E|Left1],Right) :- 
    E =< C, 
    partition(R,C,Left1,Right).
 partition([E|R],C,Left,[E|Right1]) :- 
    E > C,  
    partition(R,C,Left,Right1).
 \end{verbatim}
 \end{minipage}
  \caption{A quick-sort program}
  \label{fig:qsort-prog}

\end{figure}

%rules based on \emph{structural} orderings over covering ancestors
%have achieved highly specialized partial evaluations. 
%
%Their success
%mainly due to the fact that existing orderings capture structural
%aspects of programs and goals together with the refined dependencies
%captured by the ancestor relation.  

\noindent We now illustrate some of the ideas discussed so far and, specially,
the relevance of ancestor tracking, through an example.
%% \begin{example}
%% \label{introexample}
Our running example is the program in Figure~\ref{fig:qsort-prog},
which implements the well known quick-sort algorithm, ``{\tt qsort}'',
using difference lists. Given an initial atom of the form
{\tt qsort($List$,$Result$,$Cont$)}, where \emph{List} is a
list of numbers, the algorithm returns in \emph{Result} a sorted
difference list which is a permutation of \emph{List} and such that
its continuation is \emph{Cont}. For example, for the query
{\tt qsort([1,1,1],L,[])}, the program should compute
\texttt{L=[1,1,1]}, constructing a finite SLD tree.
Notice that, in general, if the input arguments to a program
are not sufficiently instantiated, the corresponding SLD tree can be
infinite and/or contain incomplete derivations. 
%This is the case for example for a query such as $\leftarrow
%qsort([1,1,1|Nums],L,[])$. 
%% \end{example}

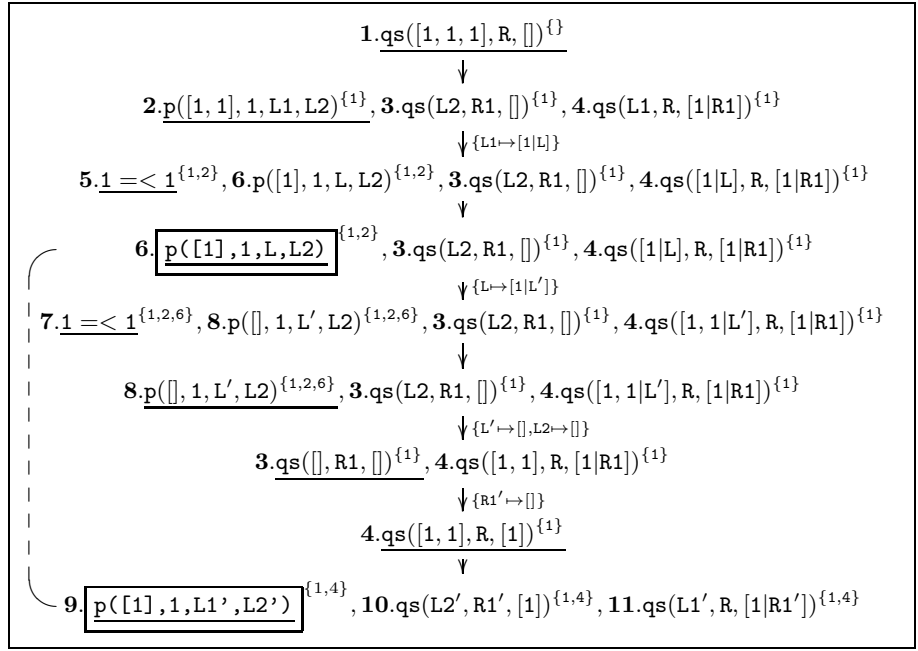
\begin{figure}[t]
%\footnotesize 
\fbox{
\begin{minipage}{11.6cm}\centering
$\xymatrix@!R=3pt{
{\bf 1.}  {\tt \underline{qs([1,1,1],R,[])^{\{\}}}} \ar[d] \\
 {\tt
 {\bf 2.} \underline{p([1,1],1,L1,L2)^{\{1\}}},{\bf 3.}
 qs(L2,R1,[])^{\{1\}},
{\bf 4.} qs(L1,R,[1|R1])^{\{1\}}}
 \ar[d]^{\tt \{L1\mapsto [1|L]\}} \\
{\tt  ~~~{\bf 5.} \underline{1 =< 1}^{\{1,2\}},
 {\bf 6.} {\tt p([1],1,L,L2)}^{\{1,2\}},{\bf 3.} qs(L2,R1,[])^{\{1\}},
   {\bf 4.}  qs([1|L],R,[1|R1])^{\{1\}}}  \ar[d] \\
{\tt  ~~~ 
 {\bf 6.} \fbox{\underline{\tt p([1],1,L,L2)}}^{\{1,2\}},{\bf 3.} qs(L2,R1,[])^{\{1\}},
   {\bf 4.}  qs([1|L],R,[1|R1])^{\{1\}}}  \ar[d]^{\tt \{L\mapsto [1|L']\}} \\
{\tt   {\bf 7.} \underline{1 =< 1}^{\{1,2,6\}},
 {\bf 8.} p([],1,L',L2)^{\{1,2,6\}},{\bf 3.} qs(L2,R1,[])^{\{1\}},{\bf 4.} qs([1,1|L'],R,[1|R1])^{\{1\}}} 
   \ar[d]\\
{\tt  
 {\bf 8.} \underline{p([],1,L',L2)^{\{1,2,6\}}},{\bf 3.} qs(L2,R1,[])^{\{1\}},{\bf 4.} qs([1,1|L'],R,[1|R1])^{\{1\}}} 
   \ar[d]^{\tt \{L'\mapsto [],L2\mapsto []\}}\\
{\tt  {\bf 3.} \underline{qs([],R1,[])^{\{1\}}},{\bf 4.} qs([1,1],R,[1|R1])^{\{1\}} }\ar[d]^{\tt \{R1'\mapsto []\}}\\
{\tt  {\bf 4.} \underline{qs([1,1],R,[1])^{\{1\}}} }\ar[d]\\
 {\tt
\ar@{--}
 `l[u]`[uuuuu]{\bf 9.} \fbox{\underline{\tt
     p([1],1,L1',L2')}}}^{\{1,4\}},
{\bf 10.} {\tt qs(L2',R1',[1])^{\{1,4\}},
  {\bf 11.} qs(L1',R,[1|R1'])^{\{1,4\}}}
}$
\end{minipage}
}\caption{Derivation with Ancestor Annotations}
\label{fig:qsort-emb}

\end{figure}

Consider now Figure~\ref{fig:qsort-emb}, which presents an incomplete
SLD derivation for our quick-sort program and the query {\tt 
qsort([1,1,1],R,[])} using a leftmost unfolding rule.  For
conciseness, predicates {\tt qsort} and {\tt partition} are
abbreviated as {\tt qs} and {\tt p}, respectively in the figure.
%
%, which has been obtained
%by using a (leftmost) unfolding rule based on the \emph{homeomorphic
%  embedding} order \cite{Leuschel:SAS98} applied over the full
% sequence. 
Note that each atom is labeled with a number (an identifier) for
future reference\footnote{By abuse of notation, we keep the same
  number for each atom throughout the derivation although it may be
  further instantiated (and thus modified) in subsequent steps. This
  will become useful for continuing the example later.}  and a
superscript which contains the list of ancestors of that atom. 
%
% We assume
% a local control rule such that further derivation steps are
% allowed as long as new selected atoms are strictly smaller---in the
% figure 
%
Let us assume that we use the homeomorphic embedding order
\cite{Leuschel:SAS98} as wqo.
%---than any previously selected atom in the same
%derivation.
%We denote by ``\emph{trace}'' the sequence of underlined
%atoms, i.e., the atoms selected along the derivation. 
If we check admissibility w.r.t.\ the full sequence of atoms, i.e., we
do not use the ancestor relation, the derivation will stop when atom
number {\bf 9}, i.e., $\tt p([1],1,L',L2')$, is found for the second
time. The reason is that this atom is greater or equal to the atom
number {\bf 6} which was selected in the third step, indeed, they are
equal modulo renaming.\footnote{Let us note that the two calls to the
  builtin predicate {\tt =<} which appear in the derivation can be
  executed since the arguments are properly instantiated.  However,
  they have not been considered in the admissibility test since these
  calls do not endanger the termination of the derivation, as we will
  discuss in Section~\ref{sec:an-unfolding-rule}.}

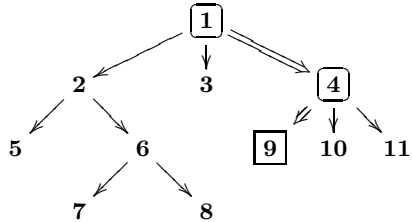
\begin{figure}[t]
\begin{center}
%\tiny 
%%  $\xymatrix@!R=0.1pt{
%%  &  & {\bf 1} \ar[dl]\ar[d]\ar@{=>}[dr] \\
%%  &  {\bf 2} \ar[dl]\ar[d] & {\bf 3} &{\bf 4} \ar@{=>}[dl]\ar[d]\ar[dr] \\
%%  {\bf 5} & 6 \ar[dl]\ar[d] & {\bf 9} &{\bf 10} &{\bf 11}\\
%%  {\bf 7} & {\bf  8} } $ %\\[1ex]
 $\xymatrix@!R=0.05pt@!C=0.05pt{
&  && \ovalbox{\bf 1} \ar[dll]\ar[d]\ar@{=>}[drr] \\
&  {\bf 2}\ar[dl]\ar[dr] & &{\bf 3} &&\ovalbox{\bf 4} \ar@{=>}[dl]\ar[d]\ar[dr] \\
 {\bf 5} && {\bf 6}\ar[dl] \ar[dr] & &\fbox{\bf 9} &{\bf 10} &{\bf 11}\\
& {\bf 7} & &{\bf  8} }$ %\\[1ex]
\end{center}
  \caption{Proof tree for the example}
  \label{fig:proof-tree}
\end{figure}

This unfolding rule is too conservative, since the process
can proceed further without risking termination
(in fact, the SLD tree
for a leftmost computation rule for the example query is finite
and thus the query can safely be fully unfolded). 
The crucial point is that the execution of atom number {\bf 9} does
not depend on atom number {\bf 6} (and, actually, the unfolding of
{\bf 6} has been already \emph{completed} when atom number {\bf 9} is
being considered for unfolding). 
In order to illustrate this, consider Figure~\ref{fig:proof-tree}
which shows the proof tree associated to this derivation. Nodes are
labeled with the numbers assigned to each atom, instead of the atoms
themselves. Note that, in order to decide whether or not to evaluate
atom number {\bf 9}, it is only necessary to check that it is not
greater or equal to atoms {\bf 4} and {\bf 1}, i.e., than those which
are its \emph{ancestors} in the proof tree.  On the other hand, and as
we saw before, if the full derivation is considered instead, as in
Figure~\ref{fig:qsort-emb}, atom {\bf 9} will be compared also with
atom {\bf 6} concluding imprecisely that the derivation may not be
safe.

Despite their obvious relevance, unfortunately the practical
applicability of unfolding rules based on the notion of covering
ancestor is threatened by the overhead introduced by the
implementation of this notion.  A naive implementation of the notion
of ancestor keeps ---for each atom--- the list of its ancestors, as it
is depicted in Figure~\ref{fig:qsort-emb} by using superscripts. This
implementation is relatively efficient in time but presents a high
overhead in memory consumption. Our experiments show that the partial
evaluator can run out of memory even for simple examples. A more
reasonable implementation maintains the proof tree as a global
structure.  In a symbolic language, this greatly reduces memory
consumption but the cost of traversing the tree for retrieving the
ancestors of each atom introduces a significant slowdown in the
partial evaluation process.  We argue that our implementation
technique is efficient in time and space, overcoming the above
limitations.

\secbeg
\section{An Efficient Implementation for Local Unfolding}
\label{sec:an-effic-impl}
\secend

In this section, we first define the notion of local computation rule.
We then introduce ASLD resolution, a modification of SLD which
incorporates ancestor stacks and which is the basis of our efficient
implementation. %  ASLD resolution in principle is not tied to local
% computation rules.
Interestingly, we then impose the local condition
to the computation rule in order to ensure accurate results for ASLD
resolution.

\subsection{A local computation rule}

Our definition of \emph{local unfolding} is based on the notion of
\emph{ancestor depth}.
\begin{definition}[ancestor depth]
  Given an SLD derivation $D=G_0,\ldots,G_m$ with $G_m= \leftarrow
  A_1,\ldots,A_k$, $k\geq 1$, the \emph{ancestor depth} of $A_i$ for
  $i=1,\ldots,k$, denoted $depth(A_i,D)$ is the cardinality of the
  ancestor relation for $A_i$ in $D$.
\end{definition}
Intuitively, the ancestor depth of an atom in a goal is the depth at
which this atom is located in the proof tree associated to the
derivation.

\begin{definition}[local computation rule]
  A computation rule $\select$ is \emph{local} if $\forall
  D=G_0,\ldots,G_n$ such that
  $G_i=\leftarrow A_{i1},\ldots,A_{i{m_i}}$ for  $i=0,..,n$, it holds that
 $depth(\select(G_i),D)\geq depth(A_{ij},D)~~~ \forall  j=1,\ldots,m_i$.
\end{definition}
Intuitively, a computation rule is local if it always selects one of
the atoms which is deepest in the proof tree for the derivation. As a
result, local computation rules traverse proof trees in a depth-first
fashion, though not necessarily left to right nor in any other fixed
order. Thus, in principle, in order to implement a local computation
rule we need to record (part of) the derivation history (i.e., its
proof tree). Note that the computation rule used in most
implementations of logic programming languages, such as Prolog, always
selects the leftmost atom.  This computation rule, often referred to
as leftmost computation rule, is clearly a local computation
rule.  Selecting the leftmost atom in all goals guarantees that the
selected atom is of maximal depth within the proof tree as it is
traversed in a depth-first fashion ---without the need of storing any
history about the derivation.

It is interesting to note that we can allow more flexible computation
rules which are not necessarily local while still ensuring termination
at the cost of no accuracy assurance. A more detailed discussion on
this will appear at the end of Section~\ref{sec:accuracy-results}.

An instrumental observation in our approach is that the proof trees
which are used in order to capture the ancestor relation
%is traversed depth-first, left-to-right, it 
can be seen as (a simplified version of) the \emph{activation
  trees}~\cite{RedDragon} used in compiler theory for representing
program executions,
%
%% In fact, the ancestor
%% subsequence in any point in time corresponds to the current
%% \emph{control word}~\cite{rozenberg-salomaa} 
by simply regarding selected atoms as procedure calls. The nodes in
such activation trees are \emph{activation records}, which contain
information about local variables, the current program counter, the
return address, etc. of the corresponding call. Nested subprogram
calls result in children activation records.
In the vast majority of programming languages, execution of a program
corresponds to traversing activation trees in a depth-first
fashion. Therefore, for efficiency, rather than maintaining the whole
activation tree in memory, run-time systems for execution of such
programming languages feature a \emph{call stack} where activation
records are stored. This call stack contains exactly the sequence of
activation records which are active at any point in time during the
execution.
This implementation strategy requires that new activation records be
added to the call stack as soon as a new subprogram is called and that
the top of the call stack is popped when the execution of a subprogram
returns.
%
%% The control word for each execution state can be seen as the set of
%% procedures whose execution has started and is not yet completed,
%% bearing a strong relation with the stack of activation records which
%% most compilers use as a run-time data structure.
%% This data structure takes normally the form of a stack, and this
%% suggests one of the central ideas of our approach: \short{  that we propose:}
%% the use of stacks for storing ancestors.
%

Our idea then is to maintain during unfolding an \emph{ancestor
  stack}, whose elements are the ancestors of a goal, instead of a
full proof tree. The advantages of this are clear: since the ancestor
stack corresponds to a single branch in the proof tree from the
current selected atom to all its ancestors in the proof tree,
maintaining it should offer significant performance improvements both
in terms of memory and time efficiency.
As in the case of control stacks, in order to compute ancestor stacks
we need to determine exactly when each ancestor should be pushed to and
popped from the ancestors stack. 
The first part is relatively simple: any resolution step requires
pushing its associated selected atom. The second part, i.e., popping
elements from the stack, is more complicated since we need to know
when the computation of the associated call (or subprogram) is
finished. In logic programming terminology this corresponds to
determining the (partial) success states for all atoms in the
derivation. In principle, success states for individual atoms are not
observable in SLD resolution, except for the top-level query. As a
result, and as we discuss below, some changes in the operational
semantics will be needed in order to make this information explicit.

Another important observation which we exploit in this paper is that
the idea of using a stack for storing the active part of a tree does
not need to be restricted to leftmost computation and it works equally
well as long as the computation rule is local. Indeed, sibling atoms,
i.e., with the same ancestor depth, can be selected in any order and
the idea of using an ancestor stack still applies.

\secbeg
\subsection{ASLD Resolution: SLD resolution with ancestor stacks}
\label{sec:lld-resolution-with}
\secend

We now propose an easy-to-implement modification to SLD resolution as
presented in Section~\ref{sec:background} in which success states for
all internal calls are observable ---and where the control word is
available at each state. We will refer to this resolution as SLD
resolution with ancestor stacks, or \emph{ASLD} for short. The
proposed modification involves 1) augmenting goals with an
\emph{ancestor stack}, which at each stage of the computation contains
the control word of the derivation, which corresponds to \emph{the
  ancestors of the next atom which will be selected for resolution},
and 2) adding pseudo-atoms to the goals used during resolution which
mark a \emph{scope} (i.e., it separates groups of atoms which are at
different depth in the proof tree). In particular, we use the
pseudo-atom $\marker$~(read as ``pop'') to indicate the end of a depth
scope, i.e., after it we move up in the proof tree. It is guaranteed
not to clash with any existing predicate name. And its purpose is
twofold: 2.1) when a mark is leftmost in a goal, it indicates that the
current state corresponds to the success state for the call which is
now on top of the ancestor stack, i.e., the call is completed,
%% and it does not need to be 
%% considered in further checks for admissibility, 
and the atom on top of the ancestor stack should be popped; 2.2) the
atoms within the scope of the leftmost mark have maximal ancestor
depth and thus a local unfolding strategy can be easily defined in the
presence of these pseudo-atoms. 

The following two definitions present the derivation rules in our ASLD
semantics.  Now, a state $S$ is a tuple of the form
$\tuple{G\gd{}\AS}$ where $G$ is a goal and $\AS$ is an ancestor stack
(or {\em stack} for short).
The stack will keep track of the
ancestor atoms that the new selected atoms need to be compared to
(by means of the wqo being used).  Thus the stack
will be instrumental in being able to stop a derivation as soon as
termination of the process can no longer be guaranteed by the
wqo being used.
To handle such stacks, we will use the usual stack operations:
\createempty, which returns an empty stack, \push$(\AS,Item)$, which
pushes \emph{Item} onto the stack $\AS$, and \pop$(\AS)$, which pops an
element from $\AS$. In addition, we will use the operation
\contents$(\AS)$, which returns the sequence of atoms contained in $\AS$
in the order in which they would be popped from the stack $\AS$ and
leaves $\AS$ unmodified.

%\begin{definition}[derivation step in ASLD]
\begin{definition}[derive]
 \label{def:nonpop-derive}
 Let $G=\;\leftarrow A_1,\ldots,\selat,\ldots,A_k$ be a goal with
 $A_1\neq \marker$.  Let $S=\tuple{G\gd{}\AS}$ be a state and $\AS$ be
 a stack.  Let $\leq_S$ be a wqo.  Let $\select$
 be a computation rule and let $\select(G)=$$\selat$ with $\selat\neq
 \marker$. Let $C=H \leftarrow B_1,\ldots,B_m$ be a renamed apart clause.
 Then $S'=\tuple{G'\gd{}AS'}$ is \emph{derived} from $S$ and $C$ via
 $\select$ if the following conditions hold: 
  \begin{eqnarray*}
    Admissible(\selat,\contents(\AS),\leq_S)\\
    \theta=mgu(\selat,H) \\
    G' \mbox{ is the goal } \leftarrow
    \theta(B_1,\ldots,B_m,\marker,A_1,\ldots,\selatpar{-1},\selatpar{+1},\ldots,A_k)\\
    AS'=\push(\AS,\selat)
  \end{eqnarray*}
%%   We say that $\selat$ is the \emph{parent} of the instance of $B_i$,
%%   $i=1,\ldots,m$, in the resolvent and in each subsequent goal where
%%   the instance originating from $B_i$ appears.
\end{definition}
The \textbf{derive} rule behaves as the one in 
Definition~\ref{def:der-step} but in addition: i) the mark $\marker$
``pop'' is
added to the goal, and ii) a  copy of $\selat$ is pushed onto the ancestor
stack. As before, the \textbf{derive} rule is non-deterministic if
several clauses in $P$ unify with the atom $\selat$.
However, in contrast to Definition~\ref{def:der-step}, this rule can
only be applied to an atom different from $\marker$ if 1) the leftmost atom in the goal is not a $\marker$
mark, and 2) the current selected atom $\selat$ together with its
ancestors do constitute an admissible sequence. If 1) holds but 2)
does not, this
derivation is stopped and we refer to such a derivation as
\emph{inadmissible} or unsafe (see Definition~\ref{def:safe-deriv}).

%\begin{definition}[pop-derivation step in LLD with ancestor stack]
\begin{definition}[pop-derive]
 \label{def:pop-derive}  Let  $G=\;\leftarrow
 A_1,\ldots,A_k$ be a goal with $A_1=\marker$.  Let
 $S=\tuple{G\gd{}\AS}$ be a state and $\AS$ be a stack.  Then
 $S'=\tuple{G'\gd{}AS'}$ with $G'=\leftarrow A_2,\ldots,A_k$ and
 $AS'=\pop(\AS)$ is \emph{pop-derived} from $S$.
\end{definition}
The \textbf{pop-derive} rule is used when the leftmost atom in the
resolvent is a $\marker$ mark. Its effect is to eliminate from the
ancestor stack the topmost atom, which is guaranteed not to belong to
the ancestors of any selected atom in any possible continuation of
this derivation.

Note that \emph{derive} steps w.r.t.\ a clause which is a fact are
always followed by a \emph{pop-derive} and thus they can be optimized
by not pushing the selected atom $\selat$ onto the stack and not
including a $\marker$~mark into the goal which would immediately pop
$\selat$ from the stack. They have been also optimized in the
implementation described in Section~\ref{sec:experiments}.  Next, we
present the following rule \textbf{derive-fact} with such an
optimization, although we do not use it for our formal developments in
Section~\ref{sec:accuracy-results}.  Indeed, its inclusion in the
semantics would require that rule \textbf{derive} is only applied if
$m>0$.
 
\begin{definition}[derive-fact]
 \label{def:nonpop-derive}
 Let $G=\;\leftarrow A_1,\ldots,\selat,\ldots,A_k$ be a goal with
 $A_1\neq \marker$.  Let $S=\tuple{G\gd{}\AS}$ be a state and $\AS$ be
 a stack.  Let $\leq_S$ be a wqo.  Let $\select$
 be a computation rule and let $\select(G)=$$\selat$ with $\selat\neq
 \marker$. Let $C= H.$ be a renamed apart fact.  Then
 $S'=\tuple{G'\gd{}AS}$ is \emph{derived} from $S$ and $C$ via
 $\select$ if the following conditions hold:
  \begin{eqnarray*}
    Admissible(\selat,\contents(\AS),\leq_S)\\
    \theta=mgu(\selat,H) \\
    G' \mbox{ is the goal } \leftarrow
    \theta(A_1,\ldots,\selatpar{-1},\selatpar{+1},\ldots,A_k)
  \end{eqnarray*}
%%   We say that $\selat$ is the \emph{parent} of the instance of $B_i$,
%%   $i=1,\ldots,m$, in the resolvent and in each subsequent goal where
%%   the instance originating from $B_i$ appears.
\end{definition}

Computation for a query $G$ starts from the state
$S_0=\tuple{G\gd{}\createempty}$.  
\short{We use $S\leadsto_P S'$ to
indicate that in program $P$ a derivation step, either \textbf{derive}
or \textbf{pop-derive} can be applied to state $S$ to obtain state
$S'$. Also, $S\leadsto^*_P S'$ indicates the transitive closure of
this relation.}
Given a non-empty derivation $D$, we denote by {\em curr\_goal(D)} and
{\em curr\_ancestors}$(D)$ the goal and the stack in the last state in
$D$, respectively.  
\short{E.g., if $D$ is the derivation $S_0\leadsto^*_P
S_n$ with $S_n=\tuple{G\gd{}\AS}$ then $curr\_goal(D)=G$ and
$curr\_ancestors(D)=\AS$. } At each step of a derivation $D$ at most
one rule, either \textbf{derive}, \textbf{derive-fact} or \textbf{pop-derive}, can be
applied.
%  depending on whether the first atom in $curr\_goal(D)$ is a
% mark $\marker$ or not.

\begin{example} Figure~\ref{fig:qsort-stack} illustrates the  ASLD
  derivation corresponding to the derivation with explicit ancestor
  annotations of Figure~\ref{fig:qsort-emb}. Sometimes, rather than
  writing the atoms themselves, we use the same numbers assigned to
  the corresponding atoms in Figure~\ref{fig:qsort-emb}.
By abuse of notation, we
    again always use the same number assigned to an atom although
    further instantiation is performed. The stack contains the
    list of atoms exactly in the instantiation state they have when
    they are pushed in the stack.
 Each step has been appropriately labeled
  with the applied derivation rule. Although rule {\em
    external-derive} has not been presented yet, we can just assume
  that the code for the external predicate {\tt =<} is available and
 has the expected behavior.

 It should be noted that, in the last state, the stack contains
 exactly the ancestors of {\tt partition([1],1,L1',L2')}, i.e., the
 atoms {\bf 4} and {\bf 1}, since the previous calls to
 \texttt{partition} have already finished and thus their corresponding
 atoms have been popped off the stack.  Thus, the admissibility test
 for {\tt partition([1],1,L1',L2')} succeeds, and unfolding can
 proceed further without risking termination.
 Indeed, the derivation can be totally unfolded, which results
  in the following (optimal) partial evaluation in which all input
  data have been satisfactorily consumed 
\[{\tt
    qsort([1,1,1],[1,1,1],[]).}\]

\begin{figure}[t]
\begin{center}
%\footnotesize
\fbox{
\begin{minipage}{11.8cm}

$\xymatrix@!R=3pt{
\tuple{ \{{\tt \underline{qs([1,1,1],R,[])}}\} \gd{} []  \ar[d]^{derive}} \\
  \tuple{ \{{\bf 2,3,4,\marker}\} \gd{} [{\tt
      qsort([1,1,1],R,[])}] \ar[d]^{derive}} \\
\tuple{ \{{\bf 5,6,\marker,3,4,\marker}\}  \gd{} [{\tt
    part([1,1],1,L1,L2),qs([1,1,1],R,[])}]} \ar[d]^{external-derive}   \\
\tuple{ \{{\bf 6,\marker,3,4,\marker}\} \gd{} [{\tt
      part([1,1],1,L1,L2),qs([1,1,1],R,[])}}] \ar[d]^{derive}  \\
 \tuple{ \{ {\bf 7,8,\marker,\marker,3,4,\marker}\} \gd{} [{\tt part([1],1,L,L2),    part([1,1],1,L1,L2),qs([1,1,1],R,[])}]\ar[d]^{external-derive}} \\
\tuple{ \{{\bf 8,\marker,\marker,3,4,\marker}\} \gd{} [{\tt part([1],1,L,L2),    part([1,1],1,L1,L2),qs([1,1,1],R,[])}]\ar[d]^{derive-fact}} \\
 \tuple{ \{{\bf \marker,\marker,3,4,\marker}\} \gd{} [{\tt part([1],1,L,L2),    part([1,1],1,L1,L2),qs([1,1,1],R,[])}]\ar[d]^{pop-derive}} \\
\tuple{ \{{\bf \marker,3,4,\marker}\} \gd{} [{\tt
     part([1,1],1,L1,L2),qs([1,1,1],R,[])}]\ar[d]^{pop-derive}} \\
 \tuple{ \{{\bf 3,4,\marker}\} \gd{} [{\tt qsort([1,1,1],R,[])}]}\ar[d]^{derive-fact} \\
 \tuple{ \{{\bf 4,\marker}\} \gd{} [{\tt qsort([1,1,1],R,[])}]}\ar[d]^{derive} \\
 \tuple{ \{{\tt part([1],1,L1',L2')},{\bf 10,11,\marker,\marker}\} \gd{} [{\tt
      qsort([1,1],R,[1]),qsort([1,1,1],R,[])}] }
}$ 
\end{minipage}}
\caption{ASLD Derivation for the example}\label{fig:qsort-stack}
\end{center}
\end{figure}
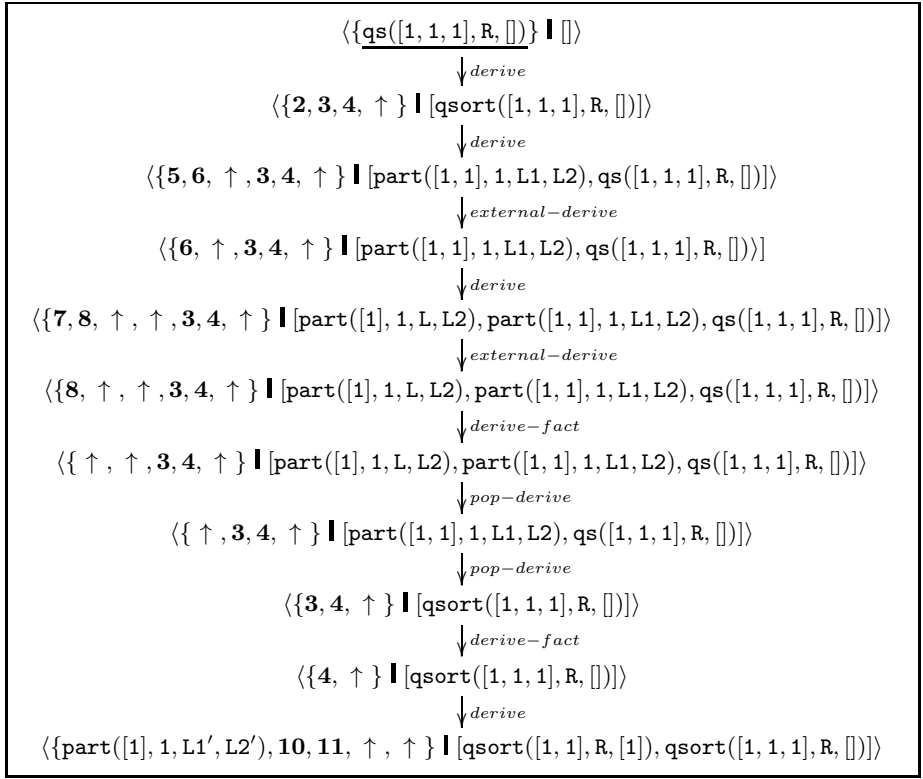
\end{example}
Finally, since the goals obtained by ASLD resolution may contain atoms
of the form $\marker$, resultants are cleaned up before being
transferred to the global control level or during the code generation
phase by simply eliminating all atoms of the form $\marker$.

It is easy to see that for each ASLD derivation $D_S$ there is a
corresponding SLD derivation $D$ with the same computed answer
and the same goal without the $\marker$~atoms.
%\mycomment{Should we write a partial correctness theorem?}
%
Such SLD derivation is the one obtained by performing the same
\emph{derive} steps (with exactly the same clauses) using the same
computation rule and by ignoring the \emph{pop-derive} steps since
goals in SLD resolution do not contain $\marker$ atoms.  We use
$\simplify(D_S)=D$ to denote that $D$ is the SLD derivation which
corresponds to $D_S$.

\subsection{Accuracy results}\label{sec:accuracy-results}

We would now like to impose a condition on the computation rule which
allows ensuring that the contents of the stack are precisely the
ancestors of the atom to be selected. 
The following notion of \emph{depth-preserving} computation
  rule allows precisely this.

\begin{definition}[depth-preserving]\label{def:depth-pres}
  A computation rule $\select$ is \emph{depth-preserving} if for each
  non-empty goal $G=\;\leftarrow A_1,\ldots,A_k$ with
  $A_1\neq\marker$, $\select(G)=\selat$ and $\marker \notin
  \{A_2,\ldots,A_{R}\}$.
\end{definition}
Intuitively, a depth-preserving computation rule always returns an
atom which is strictly to the left of the first (leftmost) $\marker$
mark. Note that $\marker$ is used to separate groups of atoms which
are at different depth in the proof tree. Thus, the notion of
depth-preserving computation rules in ASLD resolution is
\emph{equivalent} to that of local computation rules in SLD
resolution.

\begin{proposition}[ancestor stack]
\label{propo:anc-stack}
Let $D_S$ be an ASLD derivation for the initial query $G$ in program
$P$ via a \emph{depth-preserving} computation rule. Let $D$ be an SLD
derivation such that $\simplify(D_S)=D$. If,
$curr\_goal(D_S)  =  A_1,\ldots,A_n,\marker,\ldots$
and $curr\_ancestors(D_S) = \AS$, we distinguish
two cases:
\begin{itemize}
\item if $A_1 \neq \marker$, then $\contents(\AS)=Ancestors(A_i,D)$
  for $A_i \neq \marker \mbox{ for } i=1,\ldots,n$,

\item if $A_1 =\marker$, then the atom on the top of $\AS$ has no
  descendents in $curr\_goal(D_S)$ and
  $\contents(\pop(\AS))=Ancestors(A_i,D)$ for $A_i \neq \marker \mbox{
    for } i=2,\ldots,n$.
\end{itemize}

% \[ 
% \begin{array}{rcl}
% curr\_goal(D_S) & = & A_1,\ldots,A_n,\marker,\ldots$ with $A_i \neq
% \marker \mbox{ for } i=1,\ldots,n.  \\
% curr\_ancestors(D_S) &= &\AS 
% \end{array}
% \]
%Then,  $\contents(\AS)=Ancestors(A_i,D)$ for $i=1,\ldots,n$.
\end{proposition}
\begin{proof}
  The proof is by induction on the length $k$ of the ASLD derivation,
  $D_S$, of the form $S_0,\ldots,S_k$ where $S_i$, for $i=0,\ldots,k$,
  is the sequence of states corresponding to each derivation step from
  the initial state $S_0=\tuple{G\gd{}\createempty}$. To simplify the
  proof, we do not make explicit distinction between rules
  \textbf{derive} and \textbf{derive-fact}.
\begin{description}
  
\item[base case ($k=1$).] Consider the initial state
  $S_0=\tuple{G\gd{}\createempty}$ where the goal $G$ is of the form
  $\leftarrow A_1,\ldots,\selat,\ldots,A_n$, $n\geq 1$.
 Initially, all
  atoms in $G$ are different from $\marker$, i.e., $A_i \neq \marker$
  for $i=1,\ldots,n$. Therefore, we can only apply rule
  \textbf{derive} to $S_0$.  Let us assume that $\select$ is a
  \emph{depth-preserving} computation rule and $\select(G)=$$\selat$.
  Let $C=H \leftarrow B_1,\ldots,B_m$ be a renamed apart clause with
  $\theta=mgu(\selat,H)$.  The test
  $Admissible(\selat,\contents(\createempty),\leq_S)$ holds
  (otherwise the derivation step is not possible).  Then, the state
  $S_1=\tuple{G'\gd{}AS'}$ is derived from $S_0$ and $C$ where
  $G'=\theta(B_1,\ldots,B_m,\marker,A_1,\ldots,\selatpar{-1},\selatpar{+1},\ldots,A_n)$
  and $AS'=\push(\createempty,\selat)$. 

  Now, we want to prove that
  $\contents(\push(\createempty,\selat))=Ancestors(B_i,D)$,
  $i=1,\ldots,m$, for the equivalent SLD derivation $D$. Hence, we
  perform the corresponding SLD step from $\leftarrow
  A_1,\ldots,\selat,\ldots,A_m$ using the same computation rule $\cR$
  and the same clause $C$.
  In $D$, we derive the goal:
  \[\theta(B_1,\ldots,B_m,A_1,\ldots,\selatpar{-1},\selatpar{+1},\ldots,A_k)\]
  By definition of ancestor (Def.~\ref{def:ancestor-rel}), $\selat$ is
  the only ancestor of $B_i$ in $D$, $i=1,\ldots,m$. Consequently,
  $\contents(\push(\createempty,\selat))=Ancestors(B_i,D)$ holds and
  our claim follows.

\item[inductive case ($k>1$).] We decompose the ASLD derivation
  $D_{S}$ of length $k$ in two parts. The first part, $D_{S-1}$, is
  the derivation from $S_0$ to $S_{k-1}$ of length $k-1$.  The second
  part corresponds to the last ASLD derivation step from $S_{k-1}$ to
  $S_{k}$.  Let $S_{k-1}=\tuple{G_{k-1}\gd{}AS_{k-1}}$ with
  $G_{k-1}=A_1,\ldots,A_n,\marker,\ldots$ and $A_i \neq \marker$ for
  $i=1,\ldots,n$. We now distinguish two cases depending on the value
  of $n$:

%  We first apply the inductive hypothesis to the ASLD
%     derivation, $D_{S-1}$, of length $k-1$ of the form
%     $S_0,\ldots,S_{k-1}$.  Consider that $D'$ is the equivalent SLD
%     derivation obtained by  $\simplify(D_{S-1})=D'$.  Let
%     $S_{k-1}=\tuple{G_{k-1}\gd{}AS_{k-1}}$ with
%     $G_{k-1}=A_1,\ldots,A_n,\marker,\ldots$ and $A_i \neq \marker$ for
%     $i=1,\ldots,n$.  Then, $\contents(\AS_{k-1})=Ancestors(A_i,D')$
%     for $i=1,\ldots,n$. We now distinguish two cases depending on the
%     value of $n$:

\begin{itemize}

\item[\bf ($n>0$):] We first apply the inductive hypothesis to the
  ASLD derivation, $D_{S-1}$, of length $k-1$ of the form
  $S_0,\ldots,S_{k-1}$.  Consider that $D'$ is the equivalent SLD
  derivation obtained by $\simplify(D_{S-1})=D'$. Now, we perform the
  last ASLD derivation step from $S_{k-1}$.  Since $ A_1\neq \marker$,
  we can only apply rule \textbf{derive} to $S_{k-1}$.
    By assumption, $\select$ is a \emph{depth-preserving} computation
    rule. Thus, it will select an atom $\selat$ from $A_1$ to $A_n$.
    In particular, assume that $\select(G_{k-1})=$$\selat$.  Let $C=H
    \leftarrow B_1,\ldots,B_m$ be a renamed apart clause with
    $\theta=mgu(\selat,H)$. We assume that the test
    $Admissible(\selat,\contents(AS_{k-1}),\leq_S)$ holds,
    otherwise the step is not possible.  Then,
    $S_k=\tuple{G_k\gd{}AS_k}$ is derived from $S_{k-1}$ and $C$ where
\[
\begin{array}{rcl}
    G_k & =
    &\theta(B_1,\ldots,B_m,\marker,A_1,\ldots,\selatpar{-1},\selatpar{+1},\ldots,A_n,\marker,\ldots)\\
    AS_k & = & \push(\AS_{k-1},\selat)
\end{array}
\]

Now, we want to prove that $\contents(AS_k)=Ancestors(B_i,D)$, for
$i=1,\ldots,m$, for the equivalent SLD derivation $D$. Hence, we
perform the corresponding SLD step from the last goal, named $Q$, in
$D'$.  We know that $Q$ is of the form $Q=A_1,\ldots,A_n,\ldots$ since
$\simplify(D_{S-1})=D'$ and all $ A_i\neq \marker$. By using the
same local computation rule for SLD resolution, the selected
atom is also $\selat$. With the same clause $C$, we derive the goal
$\theta(B_1,\ldots,B_m,
A_1,\ldots,\selatpar{-1},\selatpar{+1},\ldots,A_n, \ldots)$.  Now, by
applying Definition~\ref{def:ancestor-rel}), the ancestors of $B_i$
are $\selat$ plus the ancestors of $\selat$ in $D'$, for
$i=1,\ldots,m$.

Finally, we proceed to put together the conclusions obtained from the
two derivations.  On one hand, we have that
$\contents(AS_{k-1})=Ancestors(A_i,D')$, $i=1,\ldots,n$. In
particular, we have that $\contents(AS_{k-1})=Ancestors(A_R,D')$ for
$i=R$. Thus, we have that:
% \[
% \begin{array}{lcl}
% \contents(AS_k) & = & \push(\AS_{k-1},\selat)\\
%                 & = & \push(Ancestors(A_R,D'),\selat) \\
%                 & = & Ancestors(B_i,D)
% \end{array}
% \]
\[
\begin{array}{lcl}
AS_k & = & \push(\AS_{k-1},\selat)\\
\contents(AS_k) & = & [\selat |AS_k] =  Ancestors(B_i,D)  
\end{array}
\] 
which proves our claim.

\item[\bf ($n=0$):] In this case, the goal is of the form
  $G_{k-1}=\marker,C_1,C_2,\ldots$. By the inductive hypothesis, we
  know that the atom on the top of $AS_k$ has no descendents in
  $curr\_goal(D_{S-1})$ and
  $\contents(\pop(AS_{k-1}))=Ancestors(C_i,D')$ for $C_i \neq \marker
  \mbox{ for } i=1,\ldots,n$.  Now, the only possibility is that
  $S_k=\tuple{G_k\gd{}AS_k}$ is \emph{pop-derived} from $S_{k-1}$ with
  $G_k= C_1,C_2\ldots$ and $AS_k=\pop(\AS_{k-1})$.  Therefore, we have
  that $\contents(AS_k)=\contents(\pop(AS_k))=Ancestors(C_i,D')$.
  Finally, in the equivalent SLD derivation step $D$ from $D'$, no
  step is performed as $\simplify$ removes the corresponding atom
  (i.e., the $\marker$ mark).  Hence,
  $Ancestors(C_i,D)=Ancestors(C_i,D')$ and the result holds.

\end{itemize}
\end{description}
\end{proof}
The above result trivially holds for leftmost unfolding which is
always depth-preserving.
The next theorem guarantees that we do not lose any specialization
opportunities by using our stack-based implementation for ancestors
instead of the more complex tree-based implementation, i.e., our
proposed semantics will not stop ``too early''. It is a consequence of
the above proposition and the results in \cite{BSM92} about wqo.

\begin{theorem}[accuracy]
\label{theo:accuracy}
Let $D$ be an SLD derivation for query $G$ in a program $P$
via a local computation rule. Let $\leq_S$ be
a wqo. If the derivation $D$ is safe w.r.t.\ 
$\leq_S$ then there exists an ASLD derivation $D_S$ for $G$ and
$P$ via
 \short{ via the equivalent to $\select$ } 
a depth-preserving
computation rule such that $\simplify(D_S)=D$.
\end{theorem}
\begin{proof}
  The proof is by contradiction.  We consider the \emph{safe} SLD derivation
  $D$ of length $k$ for $G$ via a local computation rule $\cR$.
  Trivially, the partial derivation $D'$ of length $k-1$ from $G$ to a
  goal $G'$ is safe.

  Now, the assumption is that, $D_S$, the ASLD derivation for
  $S=\tuple{G \gd{} \createempty}$ corresponding to $D$ is \emph{not}
  safe.  In particular, we consider the partial ASLD derivation,
  $D_S'$, from the state $S$ to the state $S'$, such that
  $\simplify(D_S')=D'$ and, from which a further ASLD derivation step
  for $S'$ is not safe, i.e., it would result in an inadmissible
  derivation.
  The state $S'$ is of the form $S'=\tuple{G'\gd{} AS'}$ with
  $G'=A_1,\ldots,A_n,\marker,\ldots$ and $A_i \neq \marker$, for
  $i=1,\ldots,n$. By Definition~\ref{def:depth-pres}, the
  depth-preserving computation rule can only select an atom $A_i$, for
  $i=1,\ldots,n$.
  
  Since a safe derivation step from $S'$ cannot be performed, the
  truth value of the expression: \[Admissible(A_i,contents(AS'),\leq_S)\]
  is false for any selected atom $A_i$, $i=1,\ldots,n$.  By
  Definition~\ref{def:adm}, this means that $\forall A_i,\; \exists
  B\in contents(AS'): B \leq_S A_i$. By applying
  Proposition~\ref{propo:anc-stack}, we have that the truth value of
  $Admissible(A_i,Ancestors(A_i,D'),\leq_S)$ is false as well.
  Therefore, $\forall A_i,\; \exists B\in Ancestors(A_i,D'): B \leq_S
  A_i$.
  
  Finally, since $\simplify(D_S')=D'$ and all atoms $A_i\neq \marker$,
  $G'$ is a goal of the form $A_1,\ldots,A_n,\ldots$ The equivalent
 computation rule, $\cR$, can select the same atoms $A_i$.
  However, $Admissible(A_i,Ancestors(A_i,D'),\leq_S)$ is false for all
  $A_i$, for $i=1,\ldots,n$.  Thus, the last derivation step in $D$ is
  inadmissible, hence, we have a contradiction.
\end{proof}
Note that since our semantics disables performing any further steps as
soon as inadmissible sequences are detected, not all local SLD
derivations have a corresponding ASLD derivation. However, if a local
SLD derivation is safe, then its corresponding ASLD derivation can be
found.

It is interesting to note that we can allow more flexible computation
rules which are not necessarily depth-preserving while still ensuring
termination.  For instance, consider a state
$\tuple{A_1,\ldots,A_n,\marker,\selat,\ldots\gd{} P}$ with $\marker
\notin \{A_1,\ldots,A_n\}$ and a non depth-preserving computation rule
which selects the atom $\selat$ to the right of the $\marker$ mark.
Then, rule \emph{derive} will check admissibility of $\selat$ w.r.t.\
all atoms in the stack $P$.  However, the topmost atom of $P$, say
$P_1$, is an ancestor only of the atoms $A_i$ to the left of $\selat$
but it is not an ancestor of $\selat$. The more $\marker$ marks the
computation rule jumps over to select an atom, the more atoms which do
not belong to the ancestors of the selected atom that will be in the stack,
thus, the more accuracy and efficiency we lose.  In any case, the
stack will always be an over-approximation of the actual set of
ancestors of $\selat$.

Our local unfolding rule based on ancestor stacks can be used within
any partial deduction framework, including Conjunctive Partial
Deduction (CPD) \cite{CPD:megapaper}. In principle, its use within the
CPD framework does not pose any particular difficulty and our
unfolding rule can simply be incorporated as any other strategy within
the method. Indeed, the main distinction of CPD w.r.t.\ non
conjunctive methods is on the use of an enhanced global control which
generates a set of conjunctions rather than individual atoms, while
any of the existing local control strategies can be used in
combination with such a global control.  The only requirement is that
the unfolding rule takes as input a conjunction of atoms rather than a
single atom, which is always a trivial extension. It should be noted
that some CPD examples may require the use of an unfolding rule which
is not depth-preserving to obtain the optimal specialization. As we
discuss above, we cannot ensure accuracy results (though we still have
correctness) in these cases but in turn the use of local unfolding
will improve the efficiency of the partial deduction process, as our
experimental results will show later.

\secbeg
\section{Assertion-based Unfolding  for External Predicates}
%\section{An Unfolding Rule for Full Prolog}
\label{sec:an-unfolding-rule}
\secend

Most of real-life Prolog programs use predicates which are not defined
in the program (module) being developed. We will refer to such
predicates as \emph{external}. Examples of external predicates are (1) the
traditional ``built-in'' predicates such as arithmetic operations
(e.g., \texttt{is/2}, \texttt{<}, \texttt{=<}, etc.) and basic
input/output facilities; (2) those predicates
 defined in a different module, (3) predicates written in another
language, etc.  This section deals with the difficulties that such
\emph{external} predicates pose during partial deduction and extends
our ASLD semantics to deal with them.

\subsection{The notion of evaluable atom}

When an atom $A$, such that $pred(A)=p/n$ is an external predicate, is
selected during partial deduction, it is not possible to apply the \emph{derive} rule
in Definition~\ref{def:der-step} due to several reasons. First, we may
not have the code defining $p/n$ and, even if we have it, the
derivation step may introduce in the residual program calls to
predicates which are private to the module $M$ where $p/n$ is defined.
In spite of this, if the executable code for the external predicate
$p/n$ is available, and under certain conditions, it can be possible
to fully evaluate calls to external predicates at specialization time.
We use $\exec(Sys,M,A)$ to denote the execution of atom $A$ on a logic
programming system $Sys$ (e.g., \ciao\ or SICStus) in which the module $M$,
where the external predicate $p/n$ is defined, has been loaded. In the
case of logic programs, $\exec(Sys,M,A)$ can return zero, one, or
several computed answers for $M\cup A$ and then execution can either
terminate or loop. We will use substitution sequences~\cite{ai-seq} to
represent the outcome of the execution of external predicates. A
\emph{substitution sequence} is either a finite sequence of the form
$\tuple{\theta_1,\ldots,\theta_n}$, $n\geq 0$, or an incomplete
sequence of the form $\tuple{\theta_1,\ldots,\theta_n,\bot}$, $n\geq
0$, or an infinite sequence $\tuple{\theta_1,\ldots,\theta_i,\ldots}$,
$i\in \N^*$, where ${\N}^*$ is the set of positive natural
numbers and $\bot$ indicates that the execution loops. We say that an
execution \emph{universally terminates} if
$\exec(Sys,M,A)=\tuple{\theta_1,\ldots,\theta_n}$, $n\geq 0$.

In addition to producing substitution sequences, it can be the case
that the execution of atoms for (external) predicates produces other
outcomes such as side-effects, errors, and exceptions. Note that this
precludes the evaluation of such atoms to be performed at partial
evaluation time, 
since those effects need to be performed at run-time. A clear example
of this are input/output facilities. In order to
capture the requirements which allow executing external predicates at
partial deduction time we now introduce the notion of \emph{evaluable} atom:

\begin{definition}[evaluable]
\label{def:evaluable} 
Let $A$ be an atom
  such that $pred(A)=p/n$ is an external predicate defined in module
  $M$. We say that $A$ is \emph{evaluable} in a logic programming
  system $Sys$ if $\exec(Sys,M,A)$
  satisfies the following conditions:
\begin{enumerate}
\item it universally terminates
\item it does not produce side-effects
\item it does not issue errors
\item it does not generate exceptions
\end{enumerate}
We also say that an expression $E$ is evaluable if 1) $E$ is an
evaluable atom, or 2) $E$ is a conjunction of evaluable expressions,
or 3) $E$ is a disjunction of evaluable expressions.
\end{definition}
Clearly, some of the above properties are not computable (e.g.,
termination is undecidable in the general case).
However, it is often possible to determine some \emph{sufficient
  conditions} ($SC$) which are \emph{decidable} and ensure that, if an
atom $A$ satisfies such conditions, then $A$ is evaluable.
Intuitively, a sufficient condition can be thought of as a traditional
precondition which ensures a certain behavior of the execution of a
procedure
%% at runtime 
provided they are satisfied.
Then, if this process is applied to a call corresponding to an
external predicate which is selected during partial deduction, then
that call can be executed directly at partial deduction time.
To formalize this, we propose to use the notion of \emph{evaluable
  assertion}. Basically, an evaluable assertion is a pair containing a
predicate descriptor and the sufficient conditions for its instances
to be evaluable.

\begin{definition}[correct evaluable assertion]\label{def:eval}
  Let $p/n$ be an external predicate defined in module $M$. An
  evaluable assertion $\eval{p(X1,...,Xn)}{\emph{SC}}$ is correct for
  predicate $p/n$ in a logic programming system $Sys$ if, $\forall
  \theta$:
%, the expression
%$\theta(SC)$ is evaluable, and\\
%\indent if
%$\exec(Sys,M,\theta(SC))=\tuple{\emptysubs}$ then
%$\theta(p(X1,...,Xn))$ is evaluable
 % \short{
  \begin{itemize}
  \item the expression $\theta(SC)$ is evaluable, and
  \item if $\exec(Sys,M,\theta(SC))=\tuple{\emptysubs}$ then
    $\theta(p(X1,...,Xn))$ is evaluable.
  \end{itemize}
%}

\end{definition}
In principle, assertions have to be provided manually by the supplier
of the (external) code. However, for predicates that are defined in
the source language and use only external predicates for which those
assertions are available, existing analysis tools (like those within
the \ciaopp\ system\footnote{In this system, evaluable assertions are
  called $\sf eval$
assertions.}) are able to infer them in many practical cases
(see \cite{nonleftmost-lopstr05}), as we will discuss later.

One of the advantages of using this kind of assertion is that it makes
it possible to deal with new external predicates (e.g., written in
other languages) in user programs or in the system libraries without
having to modify the partial evaluator itself. Also, the fact that the
assertions are co-located with the actual code defining the external
predicate, i.e., in the module $M$ (as opposed to being in a large
table inside the partial deduction system) makes it more difficult for the assertion
to be left out of sync when a modification is made to the external
predicate.  We believe this to be very important to the
maintainability of a real application or system library.

\begin{example}\label{ex:assertion}
Let us consider the following assertion for  the builtin predicate
  $\leq$:
%\begin{tabular}{l}
\[\tt \eval{A =< B}{(arithexpr(A),arithexpr(B))}\]
which states that if predicate \verb+=</2+ is called with both arguments
instantiated to a term of type {\tt arithexpr}, then the call is
evaluable in the sense of Definition~\ref{def:evaluable}.
In our implementation, we use the ``\emph{computational}
  assertions'' which are part of the assertion
  language~\cite{assert-lang-disciplbook-short} of \ciaopp, the Ciao
  system preprocessor~\cite{ciaopp-sas03-journal-scp}, in order to
  declare evaluable assertions.

  The type {\tt arithexpr} corresponds to arithmetic
  expressions which, as expected, are built out of numbers and the
  usual arithmetic operators. In our implementation in Ciao, the type
  {\tt arithexpr} is expressed as a unary regular logic program. This
  allows using the underlying Ciao system in order to effectively
  decide whether a term is an {\tt arithexpr} or not.
\end{example}

\subsection{The extension of ASLD resolution}

The following definition extends our ASLD semantics by providing a new
rule, \textbf{external-derive}, for evaluating calls to external
predicates.  Given a sequence of substitutions
$\tuple{\theta_1,\ldots,\theta_n}$, we define
$Subst(\tuple{\theta_1,\ldots,\theta_n})=\{\theta_1,\ldots,\theta_n\}$.
%the set of substitutions in a sequence of substitutions, i.e.,

\begin{definition}[external-derive]
 \label{def:builtin-derive} Let $Sys$ be a logic programming system.
 Let $G=\;\leftarrow A_1,\ldots,\selat,\ldots,A_k$ be a goal. Let
 $S=\tuple{G\gd{}\AS}$ be a state and $\AS$ a stack.  Let $\select$ be
 a computation rule such that $\select(G)=$$\selat$ with
 $pred(\selat)=p/n$ an external predicate from module $M$. Let $C$ be
 an evaluable assertion 
%\\ \centerline{
$\eval{p(X1,...,Xn)}{SC}.$
%}
 Then, $S'=\tuple{G'\gd{}AS'}$ is \emph{ex\-ternal-derived} from $S$ and
 $C$ via $\select$ in $Sys$ if: % the following conditions hold:
%%  1) $\sigma=mgu(\selat,p(X1,...,Xn))$,  
%%  2) $\exec(Sys,M,\sigma(SC))=\tuple{\emptysubs}$,
%%  3) $\theta\in Subst(\exec(Sys,M,\selat))$,
%%  4) $ G'$  is the goal $
%%  \theta(A_1,\ldots,\selatpar{-1},\selatpar{+1},\ldots,A_k)$,
%% 5)$ AS'=\AS$. %%  
  \begin{eqnarray*}
   \sigma=mgu(\selat,p(X1,...,Xn))   \\
   \exec(Sys,M,\sigma(SC))=\tuple{\emptysubs} \\
  \theta\in Subst(\exec(Sys,M,\selat)) \\
G'  \mbox{ is the goal }
  \theta(A_1,\ldots,\selatpar{-1},\selatpar{+1},\ldots,A_k) \\
 AS'=\AS %%  
 \end{eqnarray*}
\end{definition}
Notice that, since after computing $\exec(Sys,M,\selat)$ the
computation of $\selat$ is finished,
% there are no new pending calls. Thus,
there is no need to push (a copy of) $\selat$ into \AS\ and the
ancestor stack is not modified by the {\bf external-derive} rule.  This
%external-derive 
rule can be nondeterministic if the substitution sequence for the
selected atom $\selat$ contains more than one element, i.e., the
execution of external predicates is not restricted to atoms which are
deterministic.
The fact that $\selat$ is evaluable implies universal termination.
This in turn guarantees that in any ASLD tree, given a node $S$ in
which an external atom has been selected for further resolution, only
a finite number of descendants exist for $S$ and they can be obtained
in finite time.

\begin{example}
  Consider the \ciao\ system with the assertion in
  Example~\ref{ex:assertion} for \verb+1=<1+. Consider also the atoms
  {\bf 5} and {\bf 7}, which are of the form \verb+1=<1+, in the ASLD
  derivation of Figure~\ref{fig:qsort-emb}.  Both atoms can be
  evaluated because
  \[\exec(ciao,arithmetic,(arithexpr(1),arithexpr(1)))=\tuple{\emptysubs}\]
  This is a sufficient condition for $\exec(ciao,arithmetic,(1=<1))$
  to be evaluable. Its execution returns
  $\exec(ciao,arithmetic,(1=<1))=\tuple{\emptysubs}$.
\end{example}
In addition to the conditions discussed above which allow evaluating
atoms for external predicates at specialization time, an orthogonal
issue is that of the correctness of non-leftmost unfolding in the
presence of external predicates.
%
% As it is well known, the independence of the computation rule no
% longer holds for programs with extra logical predicates. This includes
% binding-sensitive predicates, predicates with side-effects, etc.
% %
For logic programs without impure predicates, non-leftmost unfolding
is sound thanks to the independence of the computation rule (see for
example~\cite{Lloyd87}).\footnote{However, non-deterministic unfolding
  of non-leftmost atoms can degrade efficiency.}  Unfortunately,
non-leftmost unfolding poses several problems in the context of
\emph{full} Prolog programs with \emph{impure} predicates, where such
independence does not hold anymore.
For instance, \texttt{ground/1} is an \emph{impure}  % so-called \emph{binding-sensitive}
predicate since, under LD resolution, the goal \texttt{ground(X),X=a}
fails whereas \texttt{X=a,ground(X)} succeeds with computed answer
\emph{X/a}.  Those executions are not equivalent and, thus, the
independence of the computation rule does no longer hold.
As a result, given the goal \texttt{$\leftarrow$ ground(X),X=a}, if we
allow the non-leftmost unfolding step which binds the variable
\texttt{X} in the call to \texttt{ground(X)}, the goal will succeed at
specialization time, whereas the initial goal fails in LD resolution
at run-time.  The above problem was early detected \cite{Sahlin93:ngc}
and it is known as the problem of \emph{backpropagation of bindings}.
Also \emph{backpropagation of failure} is problematic in the presence
of impure predicates.
%% There are atoms
%% $A$ for impure predicates such that $\leftarrow A,fail$ behaves
%% differently from $\leftarrow fail$.  For instance, we have to ensure
%% that failure to the right of a call to $\tt write$ does not prevent
%% the generation of the residual call to $\tt write$ nor its execution
%% at runtime. 
For instance, \texttt{$\leftarrow$ write(hello),fail} behaves
differently from \texttt{$\leftarrow$ fail,write(hello)}.

%% \begin{figure}[t]
%% \begin{minipage}{7cm}
%% \begin{verbatim}
%% :- module(main_prog,[main/2],[]).
%% :- use_module(comp,[long_comp/2],[]).
%% 
%% main(X,Y)   :-                     
%%     form_term(Type,X,Y),                       
%%     give_name(Type,X,Y).
%% 
%% form_term(functor,X,Y):- 
%%     vars(X,X'),
%%     form_functor(X',Y).
%% 
%% form_term(pred,X,Y):- 
%%     vars(X,X'),
%%     form_pred(X',Y).
%% 
%% give_name(functor,
%%           fnc(F,Args),
%%           fnc(F,_)).
%% \end{verbatim}
%% \end{minipage}
%%  \begin{minipage}{7cm}
%% %% vars3(X,Vs,Vs) :- 
%% %%     atomic(X).
%% \begin{verbatim}
%% vars(T,Vs) :- 
%%     vars3(T,[],Vs).
%% 
%% vars3(X,Vs,Vs1) :- 
%%     var(X), 
%%     insertvar(X,Vs,Vs1).
%% vars3(X,Vs,Vs1) :- 
%%     nonvar(X), 
%%     X =.. [_|Args], 
%%     argvars(Args,Vs,Vs1).
%%  
%% argvars([],Q,Q).
%% argvars([X|Xs],Vs,Vs2) :- 
%%     vars3(X,Vs,Vs1), 
%%     argvars(Xs,Vs1,Vs2).
%%  
%% insertvar(X,[],[X]).
%% insertvar(X,[Y|Vs],[Y|Vs]) :- 
%%     X == Y.
%% insertvar(X,[Y|Vs],[Y|Vs1]) :- 
%%     X \== Y, 
%%     insertvar(X,Vs,Vs1).
%% \end{verbatim}
%% %% & \fbox{\tt main\_prog} \ar[dl] \ar[dr]  \\
%% %% \fbox{\tt comp}  & & \fbox{\tt term\_typing} 
%% %% }$
%%  \end{minipage}
%%   \caption{Motivating Example}
%%   \label{fig:ex_det}
%% \end{figure}  

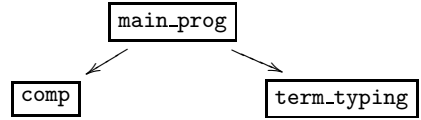
\begin{figure}[t]
\begin{minipage}{7cm}
\begin{verbatim}
:- module(main_prog,[main/2],[]).
:- use_module(comp,[long_comp/2],[]).

main(X,Y)   :- problem(X,Y), q(X).

problem(a,Y):- ground(Y),long_comp(c,Y).
problem(b,Y):- ground(Y),long_comp(d,Y).

q(a).
\end{verbatim}
\end{minipage}
\begin{minipage}{7cm}
$\xymatrix@!R=5pt@C=5pt{
& \fbox{\tt main\_prog} \ar[dl] \ar[dr]  \\
\fbox{\tt comp}  & & \fbox{\tt term\_typing} 
}$
\end{minipage}
  \caption{Motivating Example}
  \label{fig:ex_det}
\secbeg
\secbeg
\end{figure}

In order to illustrate the problem, consider the \ciao\ program in
Fig.~\ref{fig:ex_det}, which uses the impure predicate {\tt ground/1}
and whose modular structure appears to the right.  {\tt term\_typing}
is the name of the module in \ciao\ where {\tt ground/1} is defined
and predicate {\tt long\_comp/2} is imported from the user module {\tt
  comp}.  Consider a deterministic unfolding rule and the entry
declaration: ``\texttt{:- entry main(X,a).}''. The unfolding rule
performs an initial step and derives the goal {\tt
  problem(X,a),q(X)}. Then, it cannot select the leftmost atom {\tt
  problem(X,a)} because its execution performs a non deterministic
step.
In this situation, different decisions can be taken. a) We can stop
unfolding at this point. However, in general, it may be profitable to
unfold atoms other than the leftmost. Interesting computation rules
are able to detect the above circumstances and
``jump over'' the problematic atom %whose profitability criterion is not satisfied in
in order to proceed with the specialization of another atom (in this
case {\tt q(X)}). We can then decide to b) unfold {\tt q(X)} but
avoiding backpropagating bindings or failure onto {\tt problem(X,a)}.
And the final possibility c) is to unfold {\tt q(X)} while allowing
backpropagation onto \texttt{problem(X,a)}. 
However, this will require that some additional requirements hold on
the atom(s) to the left of the selected one.

There are several solutions in the literature (see,
e.g.,\cite{Leuschel:LOPSTR94,EGM97,ElviraHanusVidal02JFLP,LeuschelBruynooghe:TPLP02,LJVB04})
which allow unfolding non-leftmost atoms by avoiding the
backpropagation of bindings and failure, i.e., in the spirit of
possibility b).
Basically, the common idea is to represent explicitly the bindings by
using unification \cite{Leuschel:LOPSTR94} or residual case
expressions \cite{ElviraHanusVidal02JFLP} rather than backpropagating
them (and thus applying them onto leftmost atoms).  For our example,
by using unification, we can unfold {\tt q(X)} and obtain the
resultant {\tt main(X,a):-problem(X,a),X=a}.  This guarantees that the
resulting program is correct, but it definitely introduces some
inaccuracy, since bindings (and failure) generated during unfolding of
non-leftmost atoms are hidden from atoms to the left of the selected
one. The relevant point to note is that preventing backpropagation, by
using one of the existing methods, can be a bad idea for at least the
following reasons:
\begin{enumerate}
\item \emph{Backpropagation of bindings and failure can lead to an
    early detection of failure}, which may result in important
  speedups. For instance, if we allow backpropagating the binding {\tt
    X=a} to the left atom, we get rid of the whole (failing)
  computation for {\tt problem(b,a)} in the residual program.
\item \emph{Backpropagation of bindings can make the profitability
    criterion for the leftmost atom to hold}, which may result in more
  aggressive unfolding. In the example, by backpropagating, we obtain
  the atom {\tt problem(a,a)} which allows a deterministic computation
  rule to proceed to its unfolding.
\item \emph{Backpropagation of bindings may allow better indexing}
  by further instantiating arguments in clause heads. This is often
  good from a performance point of view (see, e.g.,
  \cite{VenkenDemoen88}).  In our example, we will obtain the clause
  head {\tt main(a,a)} with better indexing than {\tt main(X,a)}.
\end{enumerate}
The bottom line is that backpropagation should be avoided only when it
is really necessary since interesting specializations can no longer be
achieved when it is disabled.

The problems involved in and some possible solutions to non-leftmost
unfolding can be found in the
literature~\cite{Leuschel:LOPSTR94,EGM97,ElviraHanusVidal02JFLP,LeuschelBruynooghe:TPLP02}.
However, there is still ample room for improvements. In particular,
the intensive use of static analysis techniques in this
assertion-based context seems particularly promising. We are
investigating the use of the analyzers available in \ciaopp\ with this
aim, though this is outside the scope of this article.

% Our main aim in this work is to identify and characterize the
% conditions under which the possibility c) above is applicable and
% build a partial evaluation system which can effectively prove such
% conditions in order to perform backpropagation of bindings and failure
% as much as possible.

\subsection{Handling of meta-predicates}
\label{sec:meta-preds}

\begin{figure}[t]
  \centering
  \begin{lstlisting}{}
:- module(_,[p/2]).
:- use_package(library(assertions)).

p(Y,L):- findall(X,property(X),L), \+ r(Y).

:- trust pred property/1 + (eval,sideff(free)).
property(X):- q(X), \+ r(X).

q(a).  q(b).

:- trust pred r/1 + (eval,sideff(free)).
r(b).
\end{lstlisting}

  \caption{Program with meta-calls}
  \label{fig:metapreds}
\end{figure}

Though not introduced in the formalization for simplicity, our partial
evaluator can handle the usual Prolog \emph{meta-predicates}, such as
\texttt{call/1}, \texttt{findall/3}, \texttt{bagof/3}, and
\texttt{setof/3}.
Meta-predicates are characterized by receiving one or more atoms as
input. For example, \texttt{call/1} receives an atom as its only input
and \texttt{findall/3} receives a goal in its second argument
position. 
The simplest possible handling of meta-predicates consists in
residualizing all meta-calls, i.e., all calls to meta-predicates, and
transferring the atoms which appear as arguments in such meta-calls to
the global control for their subsequent partial evaluation.
For this, all meta-predicates must be declared as such and the
arguments which contain atoms must be known in advance. In the case of
\ciao\ this is done using assertions.

As a further optimization, when the atoms which appear in meta-calls
are evaluable, then rather than residualizing the meta-call, our
partial evaluator evaluates both the atom itself and also the call to
the meta predicate. This is an important optimization because partial
evaluation loses a lot of precision when unfolding is stopped and
atoms are transferred to the global control.

Another important feature of Prolog programs is negation as failure,
i.e., the \texttt{${\tt\backslash}$+/1} meta predicate.
In order to preserve the semantics of negation as failure, evaluation
of a meta-call of the form \verb- \+ A- 
% not only requires \texttt{A} itself to be evaluable but also 
requires \texttt{A} to be ground.
Therefore, at partial evaluation time a meta-call \verb- \+ A- is only
evaluated if both \texttt{A} is evaluable and ground. If this is not
the case, the meta-call is residualized and \texttt{A} is
transfered to the global control. This allows a relatively simple
handling of negation as failure where \texttt{${\tt\backslash}$+/1} is
considered as a meta predicate with the additional evaluation
requirement that its associated atom is ground.

\begin{figure}[t]
  \centering
  \begin{lstlisting}
:- module(_, [p/2] ).
:- use_package(library(assertions)).

p(A,[a]) :- \+r_1(A) .

:- trust pred r_1(_1) + (eval, sideff(free)).
r_1(b).
  \end{lstlisting}
  \caption{Partially evaluated program with meta-calls}
  \label{fig:meta-result}
\end{figure}

\begin{example}
  Figure~\ref{fig:metapreds} shows an example \ciao\ program
  containing calls to the \texttt{findall/3} meta-predicate and
  negation as failure. The \texttt{trust} assertions in \ciao\ syntax
  inform the partial evaluator that all calls to the
  \texttt{property/1} and \texttt{r/1} predicates are evaluable.

  As a result, the \texttt{findall(X,property(X),L)} meta-call is
  evaluable and can be replaced by the unification
  \texttt{L=[a]}. However, the second meta-call, i.e., \verb-\+ r(Y)-
  is residualized since \texttt{Y} remains a variable at partial
  evaluation time. The resulting program obtained by our partial
  evaluator is shown in Figure~\ref{fig:meta-result}. Since partially
  evaluated atoms are renamed, the specialized version of
  \texttt{r(Y)} has been renamed to \texttt{r\_1(Y)}. The atom
  \texttt{p(A,B)} keeps its original name since it is an exported
  predicate, in order to preserve the module interface.

\end{example}

\secbeg
\section{Experimental Results}
\label{sec:experiments}
\secend

%% In order to evaluate the  practical implications of different
%% implementations of unfolding, we have implemented in our partial deduction system a
%% number of unfolding rules.  
We have implemented in our partial evaluation system the unfolding
rule we propose, 
together with other variations in order to evaluate the efficiency of
our proposal. Our partial evaluation system has been integrated in a practical state
of the art compiler which uses global analysis extensively: the \ciao\ 
compiler and, specifically, its preprocessor
\ciaopp~\cite{ciaopp-sas03-journal-scp}.
For the tests, the whole system has been compiled using Ciao
1.13~\cite{ciao-reference-manual-1.13-shorter}.
%, with the bytecode generation option.  
All of our experiments have been performed on 
% a
%% Pentium 4 at 2.4GHz and 512MB RAM running GNU Linux RH9.0. The Linux
%% kernel used is 2.4.25.
% German's desktop
an Intel Core 2 Quad Q9300 at 2.5GHz with 1.95GB of RAM, running
% Linux 2.6.24-21. 
Linux 2.6.28-15. 

The programs used as benchmarks are indicated in the \textbf{Bench}
column.  They are classical programs often used as benchmarks for analysis
and partial evaluation of logic programs. They are described in more
detail below.
Since our proposal improves the performance of the unfolding process,
i.e., the local control, we have chosen as benchmarks programs whose
partial evaluation performs plenty of unfolding, since this allows
observing the benefits of our proposal better. In particular, three of
the benchmarks considered: \texttt{advisor3}, \texttt{query}, and
\texttt{zebra} can be fully unfolded using homeomorphic embedding with
ancestors. In the rest of the programs we provide initial queries which
are partially instantiated in order to show that our partial
evaluation system also includes global control and can partially
evaluate programs whose input data is not fully instantiated. 
Our global control is also based on homeomorphic embedding. When a new
atom is going to be specialized, we first check whether it embeds any
of the previously specialized atoms. In that case, the new atom is
generalized before being specialized by using the most specific
generalization of the new and the embedded atom. Otherwise, the new
atom is specialized as is.
For our
experiments, we use as input lists whose first part is
instantiated to integers and then the rest of the list is unknown,
i.e., just a variable, at partial evaluation time.
In the tables, we add to the name of the benchmark the number of
elements in the input list which are instantiated. For example,
\texttt{nrev\_80} should 
be interpreted as the well-known naive reverse program together with a
query which has as input a list of the form $[1,\ldots,80|T]$, with
$T$ a free variable.

The \texttt{advisor3} program is a variation of the advisor program in
the DPPD~\cite{Leuschel96:ecce-dppd} library. The \texttt{query} and
\texttt{zebra} programs are classical benchmarks for program
analysis. In particular, \texttt{query} performs a query to a small
Prolog database and \texttt{zebra} implements a simple logical puzzle.
Program \texttt{qsort}
corresponds to the quick-sort program shown in the article. The part
of the list which is instantiated is not ordered.
The \texttt{rev} benchmark is another list reversal program, but now with linear
complexity, using an accumulator. 
%% The initial query is, as before,
%% a list of 80 natural numbers. 
Finally, \texttt{permute} is a
permutation program which uses a nondeterministic deletion
predicate. 
%% It is partially evaluated w.r.t.\  a list of 6 and 7
%% elements respectively.
%
Note that two of the programs (\texttt{nrev} and \texttt{qsort}) are
partially evaluated w.r.t.\ two different input lists. The smaller of
the two corresponds to the largest possible partially instantiated
list that the partial evaluator can handle using the Relation
implementation explained below, without running out of memory.
Importantly, none of \texttt{advisor3}, \texttt{query}, nor
\texttt{zebra} can be fully unfolded using homeomorphic embedding over
the full sequence of selected atoms. Also, \texttt{nrev} and, as seen
in the running example, \texttt{qsort} are potentially not fully
unfolded if the input lists contain repetitions unless ancestors are
considered.
% 

% -*- mode: LaTeX; TeX-master: "main"; -*-

\begin{table}[t]
\centering
%% \begin{tabular}{||l||r|r|r||r|r||}\hline\hline
%% \multicolumn{1}{||c||}{}&
%% \multicolumn{3}{|c||}{\textbf{Execution Times}}&
%% \multicolumn{2}{|c||}{\textbf{Relative Speed Up}}\\\hline
\begin{tabular}{lrrrrrrrrrrrr}\hline\hline
% \multicolumn{1}{c}{}&
% \multicolumn{6}{c}{\textbf{Execution Times}}&
% \multicolumn{4}{c}{\textbf{Relative SU}}\\\hline
\multicolumn{1}{c}{}& 
\multicolumn{2}{c}{\textbf{Relation}}&
\multicolumn{2}{c}{\textbf{Tree}}
&
\multicolumn{2}{c}{\textbf{Stack}}&
 \multicolumn{2}{c}{\textbf{Rel/Stack}}&
 \multicolumn{2}{c}{\textbf{Tree/Stack}}
\\\hline
\textbf{Bench}& 
{\textbf{G}} &
{\textbf{L}} &
{\textbf{G}} &
{\textbf{L}} &
{\textbf{G}} &
{\textbf{L}} &
 {\textbf{T}} &
 {\textbf{L}} &
 {\textbf{T}} &
 {\textbf{L}} 
% {\textbf{R$_G$}} &
% {\textbf{R$_L$}} &
% {\textbf{T$_G$}} &
% {\textbf{T$_L$}} &
% {\textbf{S$_G$}} &
% {\textbf{S$_L$}} &
% %{\textbf{MEcce}} &
% {\textbf{R$_T$}} &
% {\textbf{R$_L$}} &
% {\textbf{T$_T$}} &
% {\textbf{T$_L$}} 
% & {\textbf{MEcce}} 
\\ \hline \hline
advisor3 & 0 & 103 &0 &183 &0 &151 & 0.68 &0.68 &1.21 &1.21  \\ \hline
nrev\_80 & $\infty$ & $\infty$ &38 &50622 &46 &7985& $\infty$ &$\infty$ &6.31 &6.34 \\ \hline
nrev\_43 & 12 & 912 &10 &2804 &13 &774& 1.17 &1.18 &3.58 &3.62 \\ \hline
permute6 & 4 & 526 &4 &651 &9 &453 & 1.15 &1.16 &1.42 &1.44 \\ \hline
query & 0 & 92 &0 &102 &0 &86  & 1.07 &1.07 &1.19 &1.19  \\ \hline
qsort\_80 & $\infty$ & $\infty$ &15571 &430485 &15582 &47923 & $\infty$ &$\infty$ &7.02 &8.98\\ \hline
qsort\_23 & 222 & 797 &229 &1615 &213 &566&1.31 &1.41 &2.37 &2.85\\ \hline
rev\_80 & 3 & 555 &2 &581 &2 &547 & 1.02 &1.01 &1.06 &1.06 \\ \hline
zebra & 0 & 1043 &0 &1682 &0 &1052 & 0.99 &0.99 &1.60 &1.60 \\ \hline
\hline
%  \multicolumn{7}{l}{\textbf{Overall}} & $\infty$ & $\infty$ & 6.69 & 8.21 \\ \hline\hline
\end{tabular}
\caption{Performance of Ancestor Stacks in Terms of Execution Time}
\label{tab:table1}
%\vspace*{-0.5cm}
\end{table}

In the next two tables, we compare three different implementations of
unfolding based on homeomorphic embedding with ancestors:

\begin{description}
\item[Relation] We
refer to an implementation where each atom in the resolvent is
annotated with the list of atoms which are in its ancestor relation,
as done in the example in Figure~\ref{fig:qsort-emb}.
\item[Trees] This column refers to the implementation where the
  ancestor relations of the different atoms are organized in a proof
  tree.
\item[Stacks] The column \textbf{Stacks} refers to our proposed
implementation based on ancestor stacks.

\end{description}
%
%% In the case of M-Ecce, we have not provided figures for memory
%% consumption since that would require a deep understanding of M-Ecce
%% implementation in order to make a fair comparison.
%

%\input table2_n

\subsection{Execution times}

Let us explain the results in Table~\ref{tab:table1}. Times are in
milliseconds, measuring \emph{runtime}, and are computed as the
arithmetic mean of five runs.  The partial evaluation time in each
implementation is split into two columns. The first one, labeled
\textbf{G}, shows the time taken by global control. The second one,
labeled \textbf{L}, shows the time taken by local control (i.e.,
unfolding).
The benchmarks \textbf{nrev\_80} and \textbf{qsort\_80} contain the
value $\infty$ instead of a number in the \textbf{G} and
\textbf{L} columns for \textbf{Relation}
 to indicate that the partial evaluation system
has run out of memory.  For each of these two benchmarks, we have
repeated the experiment with the largest possible initial query that
\textbf{Relation} can handle in our system before running out of
memory, i.e., \textbf{nrev\_43} and \textbf{qsort\_23}. 
%% This explains
%% that the three benchmarks are specialized w.r.t.\ two different
%% initial queries.
%
%As it can be seen in the column for relative speedups in
%Table~\ref{tab:table3}, 
\textbf{Relation} is quite efficient in time for those benchmarks it
can handle, though a bit slower than the one based on stacks. However,
and as can be seen in Table~\ref{tab:table2}, its memory
consumption is extremely high, which makes this implementation
inadmissible in practice.
Regarding \textbf{Trees}, this implementation, based on proof
trees, has good memory consumption but it is significantly slower than
\textbf{Relation} due to the overhead of traversing the tree for
retrieving the ancestors of each atom.

The last four columns compare the relative specialization times of
\textbf{Relation} and \textbf{Trees} w.r.t.\ the \textbf{Stacks}
algorithm. It should
be observed that these three alternatives are different
implementations of the same local control strategy, and that the same
global control strategy is used in all three cases. Therefore, exactly
the same residual programs are obtained in the three cases. As the
table shows (with values greater than one), \textbf{Stacks} is faster
than \textbf{Trees} in all cases.
Furthermore, \textbf{Stacks} is even faster than the implementation based
on explicitly storing all ancestors of all atoms (\textbf{Relation})
for most programs, 
while having a memory consumption comparable to (and in fact, slightly
better than) the implementation based on proof trees.
Two speedups are shown per implementation. One, named \textbf{L},
only considers the time required for local control, and the other one,
named \textbf{T}, considers the total time of global plus local
control.
The actual speedups w.r.t.\ \textbf{Trees} range from 1.06 in the case
of \texttt{rev\_80} to 8.98 \textbf{L} (7.02 \textbf{T}) in the
case of \texttt{qsort\_80}.
This variation is due to the
different shapes which the proof trees can have for the (derivations
in the) SLD tree. In the case of \texttt{rev}, the speedup is low
since the SLD tree consists of a single derivation whose proof tree
has a single branch. Thus, in this case considering the ancestor
sequence is indeed equivalent to considering the whole sequence of
selected atoms.  But note that this only happens for binary clauses.
It is also worth noticing that the speedup achieved by the
\texttt{Stacks} implementation increases with the size of the SLD
tree, as can be seen in the two benchmarks which have been
specialized w.r.t.\ different queries. The overall resulting speedup of
our proposed unfolding rule over other existing ones is significant:
over 8 times faster than our tree-based implementation.

\subsection{Memory consumption}

\begin{table}[t]
\centering
% \begin{tabular}{||l|r|r||r||r|r||}\hline\hline
\begin{tabular}{lrrrrrr}\hline\hline
\multicolumn{1}{c}{}&
\multicolumn{3}{c}{\textbf{Memory Consumption}}&
\multicolumn{3}{c}{\textbf{Relative Memory Reduction}}\\\hline
%\multicolumn{1}{||c||}{}&
%\multicolumn{3}{|c||}{\textbf{Memory Consumption}}&
%\multicolumn{2}{|c||}{\textbf{Relative Memory Reduction}}\\\hline
\textbf{Bench}& 
{\textbf{Relation}} &
{\textbf{Trees}} &
{\textbf{Stacks}} &
                   &
{\textbf{Relation}} &
{\textbf{Trees}} 
\\ \hline \hline
advisor3 & 1667260 &850612 &751112 & \hspace*{0.5cm} &  2.22 &  1.13\\ \hline
nrev\_80 & mem &1076384 &944936 & \hspace*{0.5cm} & $\infty$ &  1.14\\ \hline
nrev\_43 & 56255068 &1103980 &1041490 & \hspace*{0.5cm} &  54.01 &  1.06\\ \hline
permute\_6 & 23361920 &1959004 &1431976 & \hspace*{0.5cm} &  16.31 &  1.37\\ \hline
query & 2764368 &8064 &7520 & \hspace*{0.5cm} &  367.60 &  1.07\\ \hline
qsort\_80 & mem &5660460 &5038540 & \hspace*{0.5cm} & $\infty$ &  1.12\\ \hline
qsort\_23 & 11130584 &630048 &598212 & \hspace*{0.5cm} &  18.61 &  1.05\\ \hline
rev\_80 & 2552524 &144264 &139076 & \hspace*{0.5cm} &  18.35 &  1.04\\ \hline
zebra & 26819712 &107760 &101280 & \hspace*{0.5cm} &  264.81 &  1.06\\ \hline
\hline
 \multicolumn{5}{l}{\textbf{Overall}}& $\infty$& 1.15\\ \hline\hline
\end{tabular}
\caption{Performance of Ancestor Stacks in terms of Memory Consumption}
\label{tab:table2}
\end{table}

We have also studied the memory required by the unfolding process.
Let us briefly discuss the figures depicted in Table~\ref{tab:table2}
which represent, in number of bytes, memory consumption. It has been
measured at each derivation step during the construction of the ASLD
trees.
%%  by subtracting, for each memory area , the
%% current level from the level at the beginning of the construction of
%% the ASLD tree. 
At each step, the resulting numbers for all memory areas (stack, heap,
etc.)  have been added and then compared to the previous maximum
value, taking always the larger of the two, thus computing the high
water mark, i.e., the maximum memory required to perform unfolding. The
figures show, for each benchmark, the high water mark minus the memory
already in use when the construction of the SLD tree was started. 
In order to make these numbers closer to the actual memory used, 
garbage collection has remained enabled during the different
experiments. In order to make memory figures comparable, we force
garbage collection just before starting partial evaluation of each
benchmark. 
%
%The first set of column presents the absolute memory consumption,
%whereas the second set presents relative figures taking the usage of
%the \textbf{Stacks} algorithm as unit. 

In the last row, labeled \textbf{Overall}, we summarize the results
for the different benchmarks using a weighted mean, which places more
importance on those benchmarks with relatively larger unfolding
figures.  We use as weight for each program its actual unfolding
time/memory.  We believe that this weighted mean is more informative
than the arithmetic mean, as, for example, doubling the speed in which
a large unfolding tree is computed is more relevant than achieving
this for small trees.

As Table~\ref{tab:table2} shows, 
%and as for the case of execution time, 
the \textbf{Stacks} algorithm presents lower consumption than
either of the two other algorithms studied for any of the programs. 
It can be seen that the
amount of memory required by the \textbf{Relation} algorithm precludes it from
its practical usage. Regarding the \textbf{Stacks} algorithm, not only
it is significantly faster than the implementation based on trees.
Also it provides a relatively important reduction (1.15 overall,
computed again using a weighted mean) in memory consumption over
\textbf{Trees}, which already has a good memory usage.

Altogether, when the results of Table~\ref{tab:table1} and
Table~\ref{tab:table2} are combined, they provide evidence that our
proposed techniques allow significant speedups while at the same time
requiring somewhat less memory than tree based implementations and
much better memory consumptions than implementations where the
ancestor relation is directly computed. This suggests that our
techniques are indeed effective and can contribute to making partial
evaluation a practical tool.

\subsection{Comparison with Ecce. Specialization Quality.}

% -*- mode: LaTeX; TeX-master: "main"; -*-

\begin{table}[t]
\centering
%% \begin{tabular}{||l||r|r|r||r|r||}\hline\hline
%% \multicolumn{1}{||c||}{}&
%% \multicolumn{3}{|c||}{\textbf{Execution Times}}&
%% \multicolumn{2}{|c||}{\textbf{Relative Speed Up}}\\\hline
\begin{tabular}{lrrrrrrrrrrrrr}\hline\hline
% \multicolumn{1}{c}{}&
% \multicolumn{2}{c}{\textbf{Relation}}&
% \multicolumn{2}{c}{\textbf{Tree}}&
% \multicolumn{3}{c}{\textbf{Quality}}
% \\\hline
\textbf{Bench}& 
{\textbf{Ecce$_L$}} &
{\textbf{Ecce$_G$}} &
{\textbf{Orig}} &
{\textbf{Stacks}} &
{\textbf{Ecce}} &
{\textbf{O/S}} &
{\textbf{O/E}} &
{\textbf{E/S}} 
%Benchmark & Glob & Loc & Orig & CiaoPP & Ecce & O/S    & O/E   & E/S    \\ \hlin% {\textbf{R$_G$}} &
% {\textbf{R$_L$}} &
% {\textbf{T$_G$}} &
% {\textbf{T$_L$}} &
% {\textbf{S$_G$}} &
% {\textbf{S$_L$}} &
% %{\textbf{MEcce}} &
% {\textbf{R$_T$}} &
% {\textbf{R$_L$}} &
% {\textbf{T$_T$}} &
% {\textbf{T$_L$}} 
% & {\textbf{MEcce}} 
\\ \hline \hline

advisor3  &  0 &30 & 1149 & 1119   & 1042 & 1.03   & 1.10  & 0.93   \\ \hline
nrev\_80   & 19310 & 1297700 & 1005 & 30     & 64   & 33.50  & 15.70 & 2.16   \\ \hline
nrev\_43  &   2910 & 105600  &864  & 41     & 102  & 21.07  & 8.50  & 2.49   \\ \hline
permute\_6 & 40  & 20  &934  & 301    & 620  & 3.10   & 1.51  & 2.06   \\ \hline
query     & 20& 90 &  1106 & 55     & 570  & 20.11  & 1.94  & 10.32  \\ \hline
qsort\_80  & 85300&269070&  1178 & 15     & 17   & 78.53  & 70.14 & 1.11   \\ \hline
qsort\_23  & 260 & 900  & 978  & 34     & 34   & 29.12  & 29.12 & 1.00   \\ \hline
rev\_80    & 10& 730  & 1132 & 712    & 704  & 1.59   & 1.61  & 0.99   \\ \hline
zebra     & 170 & 300   & 6      & 1069 & 384.90 & 2.30  & 167.00 \\ \hline
\hline
\end{tabular}
\caption{Comparison with Ecce. Specialization Quality. }
\label{tab:tableEcce}
%\vspace*{-0.5cm}
\end{table}

Finally, in Table~\ref{tab:tableEcce}, we want to compare our
implementation with that of a state-of-the-art partial evaluator and
see the quality of the specialized programs.  To do so, we have also
measured the time that it takes to process the same benchmarks using
Leuschel's Ecce~\cite{Leuschel96:ecce-dppd} system. For this, we have
used the compiled version available at
http://www.stups.uni-duesseldorf.de/\~{}asap/asap-online-demo/meccedownloads
and run the experiments on the same machine.  These execution times
are provided in columns
 {\textbf{Ecce$_L$}} and {\textbf{Ecce$_G$}} which show,
respectively, the  time taken by local and the global control in
Ecce. When compared with {\textbf{L}} and {\textbf{G}} in
Table~\ref{tab:table1} for the stack implementation, the results
provide evidence that our proposed stack-based implementation compares
quite well with state of the art systems as regards specialization
times. Indeed, the specialization times using our stack-based
implementation are considerably smaller for all benchmarks with high
local control times. In those benchmarks in which Ecce is faster than
the \textbf{Stacks} implementation, it is due to the unfolding rules
not being identical which results in Ecce performing fewer unfolding
steps.
%% This is because the unfolding rules are not
%% exactly the same. 
Note that performing less unfolding may lead to less specialized
programs, which are often less efficient.

The next columns aim at evaluating the quality of the specialized
programs in Ecce and in our system by comparing their runtimes with
those of the original programs. We have chosen sufficiently large
input data and run the original program (column \textbf{Orig}), the
specialized one by our system (column \textbf{Stacks}) and the
specialized one by Ecce (column \textbf{Ecce}) on the same data and
the same number of times and show the aggregated runtime. The last
three columns show the speedup achieved for each benchmark. In
particular, \textbf{O/S} and \textbf{O/E} show, respectively, the
speedup of Stacks and Ecce w.r.t. the original program and
\textbf{E/S} compares Ecce against Stacks. It should be observed that
in all cases the specialized programs in both Ecce and Stacks are more
efficient than the original ones and in most cases the gain is
significant. The cases in which Stacks performs better than Ecce
(e.g., \texttt{query} and \texttt{zebra}) are because we can fully
unfold them in Stacks while Ecce stops the specialization
earlier. Hence, the gain is much larger. It is also important to
notice that in the example \texttt{advisor}, the specialization
obtained by Stacks also performs more unfolding steps than the one in
Ecce. In this case, such additional unfolding results in an unneeded
over-specialization which increases the size of the residual program
and leads to a less efficient execution.

\section{Related Work and Conclusions}\label{sec:disc-future-work}

The development of powerful unfolding rules has received considerable
attention during the last years \cite{LeuschelBruynooghe:TPLP02}.  The
most successful techniques to date are based on two fundamental ingredients:
\begin{itemize}
\item the use of a wqo which can be used to guarantee
  termination while achieving very powerful unfoldings,
\item structuring the atoms already visited in each derivation in a
  tree rather than using an unstructured collection, such as a set.
\end{itemize}
Among the  well-quasi orderings, the
\emph{homeomorphic embedding} \cite{Kru60,LeuschelBruynooghe:TPLP02}
has proved 
to be very powerful in practice. Regarding the structure to use for
visited atoms, the notion of \emph{ancestors} seems to be the best one
since it guarantees termination while allowing transformations which
are strictly more powerful than those achievable if unstructured
collections are used.

The use of ancestors for refining sequences of visited atoms was
proposed early on by \cite{BSM92} and significant effort has been
devoted to improve the implementation of
ancestors~\cite{MartensDeSchreye:jlp95b}.  However, the combination of
wqo and ancestors happens to be very inefficient in
practice. This is mainly due to the fact that dependency information
has to be maintained for the individual atoms in each derivation.  In
principle, the use of ancestors should not only allow more powerful
transformation but also speed up unfolding since it reduces the length
of sequences for which admissibility has to be checked.
Unfortunately, maintaining such information about ancestors during the
generation of SLD trees introduces a costly overhead which can
eliminate the theoretical efficiency gains.
 
In this work we have proposed ASLD resolution, a novel extension over
the SLD semantics to incorporate ancestor stacks which can be used as
a basis for the \emph{efficient} generation of (incomplete) SLD trees
during partial deduction in combination with wqo. The
main features of the implementation technique and extensions  that we
propose  for the ancestor-based local 
unfolding rule, based on ASLD resolution, are: (1)
it is parametric w.r.t.\ the wqo of interest; (2) it can
handle logic programs with builtins; (3) it is guaranteed to always
provide finite trees; (4) it is very easy to implement since the
ancestor information is simply stored using a stack; (5) it provides a
very efficient implementation of ancestor information; (6) if certain
conditions are imposed on the computation rule, then it is as accurate
as standard (more inefficient) unfolding rules based on ancestors.
%
%In the particular case of partial evaluation of Prolog programs, it
%makes sense to restrict ourselves to leftmost unfolding.  This is
%known as LD resolution. 
Note that, as it is the case with unfolding rules based on traditional
SLD resolution, our semantics can be used in combination with a
determinacy check which may decide to stop unfolding even if
termination is guaranteed whenever too many alternative,
non-deterministic, branches are generated in the SLD tree.

The unfolding rule proposed in this work has been implemented in the
\ciaopp\ system \cite{ciaopp-sas03-journal-scp}, the preprocessor of
the \ciao~programming language.  Experimental results are 
%% very
promising: they provide evidence that our proposed techniques allow
significant speedups while at the same time requiring somewhat less
memory than tree-based implementations and much better memory
consumptions than implementations where the ancestor relation is
directly computed.
%
%% Though it can be argued that specialization time is not as relevant as
%% execution time, the takeup of partial deduction techniques also
%% depends on the power of existing systems and tools and we hope our
%% proposal can contribute to increasing this power.
Though specialization time is obviously not as critical as execution
time, being able to perform powerful specializations in reasonable
time can only contribute to the practical takeup of partial deduction
techniques.

As for future work, we plan to incorporate in our partial evaluator 
(embedded in \ciaopp) the extensions needed to perform Conjunctive
Partial Deduction and to 
investigate whether local unfolding can be successfully used in 
% these examples.  
this context. 
We are also investigating
additional solutions for the problems involved in non-leftmost
unfolding for programs with extra logical predicates beyond those
presented in the
literature~\cite{Leuschel:LOPSTR94,EGM97,ElviraHanusVidal02JFLP,LeuschelBruynooghe:TPLP02}.
In particular, the intensive use of static analysis techniques in this
context seems particularly promising. In our case we can take
advantage of the fact that our partial deduction system is integrated
in \ciaopp, which includes extensive program analysis facilities.
A first step in this direction has been taken in
\cite{nonleftmost-lopstr05} by using backwards analysis to infer
purity assertions which determine when a non-leftmost step is safe in
the presence of impure predicates.

\secbeg
\subsection*{Acknowledgments}
\secend

\begin{small}
%% \noindent
%%   \textbf{Acknowledgments:} 
  We gratefully acknowledge the anonymous referees for many useful
  comments and suggestions that helped to improve this article.  This
  work was funded in part by the Information Society Technologies
  program of the European Commission, under the Future and Emerging
  Technologies IST-231620 {\em HATS} project and the IST-215483 {\em
    SCUBE} project, by the Spanish Ministry of Science and Innovation
  under the TIN-2008-05624 {\em DOVES} and HI2008-0153 projects, and
  by the Madrid Regional Government (CM) under the S-0505/TIC/0407
  \emph{PROMESAS} project.
%
% Manuel Hermenegildo is also supported by the
%   Prince of Asturias Chair in Information Science and Technology at
%   UNM.
\end{small}

\end{document}